\documentclass[runningheads]{llncs}
\usepackage{graphicx}
\usepackage{xcolor}
\usepackage[utf8]{inputenc}
\usepackage{booktabs} % For formal tables
\usepackage{amsfonts}
\usepackage{array}
\usepackage{graphicx}
\usepackage{hyperref}

\usepackage{amssymb}  
\usepackage{amsmath}
\usepackage[bottom,flushmargin,hang,multiple]{footmisc}

\usepackage{etoolbox}
\newtoggle{proof_thm_1}
\newtoggle{proof_thm_1_sketch}

\newtoggle{proof_thm_2}
\newtoggle{proof_thm_2_sketch}

\newtoggle{proof_thm_3}
\newtoggle{proof_thm_3_sketch}

\newtoggle{proof_thm_4}
\newtoggle{proof_thm_4_sketch}

\newtoggle{proof_thm_5}
\newtoggle{proof_thm_5_sketch}

\newtoggle{proof_thm_6}
\newtoggle{proof_thm_6_sketch}

\newtoggle{proof_thm_7}
\newtoggle{proof_thm_7_sketch}

\newtoggle{proof_thm_8}
\newtoggle{proof_thm_8_sketch}

\newtoggle{proof_thm_9}
\newtoggle{proof_thm_9_sketch}

\newtoggle{proof_thm_10}
\newtoggle{proof_thm_10_sketch}

\newtoggle{proof_lemma_1}

\toggletrue{proof_thm_4_sketch}
\toggletrue{proof_thm_6_sketch}
\toggletrue{proof_thm_8}
\togglefalse{proof_thm_1}
\togglefalse{proof_thm_2}
\togglefalse{proof_thm_3}
\togglefalse{proof_thm_5}
\togglefalse{proof_thm_7}
\togglefalse{proof_thm_9}
\togglefalse{proof_thm_10}
\togglefalse{proof_lemma_1}

\newcommand{\D}{\mathbf{D}}
\newcommand{\DI}{\mathbf{D}_{-i}}

\newcommand{\TAPStar}{TAP$^+$}

\begin{document}
\title{Social Cost Guarantees in Smart Route Guidance
%Strategyproof Traffic Assignment in Smart Traffic Control Systems
}
\author{Paolo Serafino\inst{1} \and
Carmine Ventre\inst{2} \and
Long Tran-Thanh\inst{3} \and
Jie	Zhang\inst{3} \and
Bo An\inst{4} \and
Nick Jennings\inst{5}
}
\institute{Gran Sasso Science Institute, Italy
\email{paolo.serafino@gssi.it}\\
\and
University of Essex, UK, \email{c.ventre@essex.ac.uk}
\and
University of Southampton, UK, \email{\{ltt08r,jie.zhang\}@soton.ac.uk	}
\and 
Nanyang Technological University, Singapore,
\email{boan@ntu.edu.sg}
\and
Imperial College London, UK, \email{n.jennings@imperial.ac.uk}}

%\author{Paper \#}%
%\author{First Author\inst{1}\orcidID{0000-1111-2222-3333} \and
%Second Author\inst{2,3}\orcidID{1111-2222-3333-4444} \and
%Third Author\inst{3}\orcidID{2222--3333-4444-5555}}
\maketitle              % typeset the header of the contribution
\begin{abstract}
We model and study the problem of assigning traffic in an urban road network infrastructure. 
In our model, each driver submits their intended destination
and is assigned a route to follow %from their current location 
%to their destination 
that minimizes the social cost (i.e., travel distance of all the drivers).
We assume drivers are strategic and try to manipulate the system (i.e., misreport their intended destination and/or deviate from the assigned route) if they can reduce their travel distance by doing so.
Such strategic behavior is highly undesirable as it can lead to an overall suboptimal traffic assignment and cause congestion. % and hence increase the travel costs of other agents.
To alleviate this problem, we develop moneyless mechanisms that are resilient to
manipulation by the agents and offer provable approximation guarantees on the social cost obtained by the solution.
%The novelty of our study is twofolds. On one hand we study mechanisms that do not require monetary incentives, and thus, can be applied to many scenarios where payment or monetary based penalty collection is not available.
%On the other hand, we model the reaction of an agent to the allocation she is assigned, and we obtain theoretical guarantees that such an allocation will be implemented by the agents.
%
%TODO: add the results here.
%In particular, we provide lower bounds on the approximability of both deterministic and randomized strategyproof algorithms for the traffic assignment problem.
We then empirically test the mechanisms studied in the paper, showing that they can be effectively used in practice in order to compute manipulation resistant traffic allocations.

\end{abstract}

% !TEX root = ijcai18.tex
\section{Introduction}

%\begin{itemize}
%\item traffic control's importance in future smart cities
%\item goal is to avoid congestion and reduce travelling time
%\item typical smart routing system takes user travel request and suggest travel paths
%\item users can be strategic and intentionally submit incorrect data
%\item show example
%\item Research challenge: design a mechanism that forces them to truthfully report
%\item this leads to mechanism design - introduce the concept of strategyproofness
%\item existing mechanism design literature typically focuses on monetary based mechanism - which is a big restriction
%\item give an example of non-monetary based system
%\item research question: can we do this without money?
%\item Our paper: blablabla
%\end{itemize}

Recent years have witnessed increasing interest in the development of efficient traffic control systems \cite{osorio2015urban,leontiadis2011effectiveness,djahel2013adaptive}. This is motivated by the significant negative impact on the quality of life of both road users and residents caused by heavy traffic congestion levels in large cities such as London, Beijing, and Los Angeles. 
Indeed, heavy congestion is known to be a major cause of air and noise pollution, which are widely recognized as the main cause of many health issues~\cite{Kryzanowski2005,Stanfeld2003}. 
Adding to this is the economic cost associated with the large amount of time spent in traffic jams, which reduces the productivity of the economy~\cite{goodwin2004economic}.
Moreover, the situation is expected to become significantly worse in the future when the population, and thus the traffic flow, in large cities will be much bigger than at present.
Unfortunately, conventional traffic control systems have proven unable to efficiently decrease congestion levels, as they are not designed to be adaptive to the dynamics of city traffic, which changes over space and time. 
On the other hand, it has been shown~\cite{raphael2015goods,levin2017optimizing} that by putting some sort of intelligence/smartness into traffic control systems, we can make them adapt to the changes of the traffic flow. 
A key objective within these smart traffic control systems is to address the so-called \emph{traffic assignment problem} (TAP), in which mobile agents (i.e., typically drivers) declare their intended destination to the system, perhaps via their satellite navigation systems, and are then assigned a route to follow, in such a way that some objective function of the overall traffic flow in the system is optimized (i.e., minimizing the total traveled distance or maintaining an efficient traffic load balance).
As these agents are typically self-interested and strategic (i.e., they try to maximize their own utility, disregarding whether this is detrimental to the global optimization goal), they may manipulate the system whenever they can benefit from doing so~\cite{levin2017optimizing,vasirani2012market}.
%example
%For example, consider an agent who wants to travel from point $A$ to point $B$.
%If she truthfully reports $B$ as her preferred destination, she will be allocated path $P_1$ by the traffic assignment algorithm.
%For the sake of the argument, let us assume that when the agent wants to travel to point C, she is assigned a path $P_2$ that goes through point B via a shorter route (i.e., it might be that, for load balancing issues, agents who travel farther away are assigned the shortest routes available).
%Therefore, the agent is incentivized to report $C$  and to reach her intended destination $B$ through path $P_2$. 
%
This kind of opportunistic behavior is highly undesirable as it will increase the total social cost (i.e., decreasing the total load balance or increasing the total congestion level).
As such, incentivizing agents not to be strategic is a key design objective of these traffic assignment systems~\cite{raphael2015goods,levin2017optimizing,vasirani2012market}.
Given this, we focus on \emph{strategyproof} %~\cite{ProcacciaT13,resource_augmentation}
TAP mechanisms, which guarantee that it is in the agent's best interest to always report her true destination and follow the assigned route.
Furthermore, we assume that money transfers between the mechanism and the agents are not available.
This is a common assumption in many domains \cite{ProcacciaT13} that will facilitate the likely real-world deployment of the system by lowering set up costs (i.e., avoiding the construction of tolling booths).
The remainder of the paper is organized as follows.
In Section 2 we discuss related works. In Section 3 we introduce our model for TAP and prove that Pareto optimal allocations theoretically guarantee that  agents will follow their assigned paths (Theorem \ref{thm:pareto_opt}). 
%Paolo version: We first prove that Pareto optimal allocations theoretically guarantee that risk-neutral agents will follow their assigned paths (Theorem \ref{thm:pareto_opt}). In light of this, we identify the following set of desiderata for mechanisms in the traffic assignment problem: (\emph{i}) strategyproofness, (\emph{ii}) Pareto optimality, and (\emph{iii}) non-bossiness (Section \ref{sec:model}). 
%
We then move to study deterministic (Section \ref{sec:deterministic_mechs}) and randomized (Section \ref{sec:randomized_mechs}) Pareto optimal mechanisms for our problem.
We show that the approximation ratio of \textit{deterministic} strategyproof mechanisms is lower bounded by 3 (Theorem \ref{thm:deterministic_LB}), while the Serial Dictatorship mechanism can achieve an upper bound of $2^n - 1$ and it is Pareto-optimal and non-bossy (Theorems \ref{thm:apx_SD} and \ref{thm:tightness_SD}), where $n$ is the number of agents (Theorems \ref{thm:apx_SD} and \ref{thm:tightness_SD}).
Furthermore, if we require non-bossiness and Pareto optimality, we are able to close this approximation ratio gap by showing that the Bipolar Serial Dictatorship mechanism is the \emph{only} strategyproof mechanism.
%Paolo: then focus on deterministic mechanisms and prove that there exist no strategyproof mechanisms having approximation ratio lower than 3 (Theorem \ref{thm:deterministic_LB}). On the positive side, we show that, under some mild conditions, the Serial Dictatorship mechanism can achieve our three desiderata with a $2^n - 1$ approximation ratio (Theorems \ref{thm:apx_SD} and \ref{thm:tightness_SD}), where $n$ is the number of agents.
%For a subset of instances of TAP we are able to close this exponential gap by characterizing Bipolar Serial Dictatorship (a slight extension of Serial Dictatorship, yielding the same approximation ratio) as the \emph{only} strategyproof mechanism.
%We conjecture such a characterization holds for the general TAP as well.
%
For \textit{randomized} mechanisms, we show that the approximation ratio is lower bounded by $\frac{11}{10}$ (Theorem \ref{thm:randomizedLB}). In addition, the Random Serial Dictatorship mechanism can achieve an  $n$-approximation (Theorems 8 and 9), while still preserving the desired properties of Pareto-optimality and non-bossyness.
%For a subset of instances of TAP, we prove (Theorem 10) a $\Omega(\sqrt{n})$ approximation lower bound for randomized strategyproof mechanisms.
%Paolo: In order to obtain better approximation guarantees, we next turn to randomized mechanisms. We prove that no $\alpha$-approximate randomized mechanism, with $\alpha \leq \frac{11}{10}$, is strategyproof (Theorem \ref{thm:randomizedLB}). On the positive side, we show that the randomized version of Serial Dictatorship can achieve tight $n$ approximation (Theorems 8 and 9), while still preserving (the randomized counterpart of) our 3 desiderata, namely: (universal) strategyproofness, (ex-post) Pareto optimality, and non-bossiness.
In addition to these theoretical results, we present an extensive experimental evaluation on traffic networks generated from real road network data, which show how the mechanisms studied in the paper provide good performance in practice, despite the high theoretical worst case approximation guarantee.  

\noindent
Full proofs and definitions can be found in the Appendix.

%\noindent \textbf{Roadmap}: %The remainder of the paper is organized as follows.
%Firstly, we briefly discuss related work. Secondly, we formalize TAP and give some preliminary definitions and our key result on Pareto optimality.
%Thirdly, we show theoretical approximation ratio bounds of deterministic mechanisms and randomized mechanisms.
%Lastly, we present our experimental results and conclude.

% !TEX root = infocom17_main.tex
\section{Related Work}\label{sec:related_work}
%\begin{itemize}
%\item We follow the model proposed by Beckmann et al. \cite{beckmann1956studies}
%\item resource augmentation is in \cite{resource_augmentation}
%\item approximate mechanism designe without money is in \cite{ProcacciaT13}
%\item \cite{Brenner2010}
%\end{itemize}
%

There is a large body of literature on traffic network modelling and assignment~\cite{beckmann1956studies,skabardonis1997improved,coogan2015compartmental,daganzo1994cell}.
However, these works typically ignore the strategic behaviour of participating agents. %, and thus, are not suitable to provide strategyproofness.
Nevertheless, they can be useful to model the underlying traffic network in our work.
In particular, we follow the widely used traffic model proposed in~\cite{beckmann1956studies}.

To tackle the strategic behaviour of the agents, several researchers have suggested employing mechanism design with money and auction theory for traffic control~\cite{raphael2015goods,levin2017optimizing,vasirani2012market,Brenner2010}.
These works typically rely on the computation of the VCG auction in order to assign vehicles to paths. 
However, they require monetary incentives, and typically focus on a local control level, such as intersection management (as VCG is typically computationally hard, and thus, not readily scalable~\cite{conitzer2006failures}). 

A number of researchers have focused on mechanism design without money~\cite{ProcacciaT13,resource_augmentation}.
However, none of these mechanisms can be easily applied to the traffic assignment problem, as they do not take into account the features of the underlying traffic network structure.
As we will show, TAP bears some resemblance to the problem of assigning indivisible objects \cite{DBLP:journals/jet/BogomolnaiaDE05,Svensson1999,Filos-RatsikasF014}, although these results are not directly applicable to our scenario. Indeed TAP has a much more complex structure (mainly due to the underlying traffic network topology) which traditional assignment mechanisms fail to address. % and we will 
%Specifically, we will investigate the usage of the Serial Dictator mechanism discussed by Caragiannis \emph{et al}.~\cite{resource_augmentation}.

% !TEX root = ijcai18.tex
\section{Model and Preliminary Definitions}\label{sec:model}
A \emph{traffic assignment problem} (TAP) consists of a set of agents $A=\{a_1, \ldots, a_n \}$ and a road network infrastructure, represented as a directed graph $G=(V,E)$, where: (\emph{i}) $V = \{v_1,\ldots, v_{|V|}\}$ is the set of nodes representing the junctions of the road network infrastructure; and (\emph{ii}) $E\subseteq V\times V$ is the set of directed edges representing one-way road segments.
Each edge $e\in E$ has a \emph{capacity} $c : E \rightarrow \mathbb{N}^+$, which determines the maximum number of agents that can travel through the edge at any given time, and a \emph{weight function} $w: E \rightarrow \mathbb{R}^+$ which represents the cost incurred by the agent traveling through the edge (i.e., travel distance). Furthermore, each edge is associated to a \emph{transit time} $\tau: E \rightarrow \mathbb{Z}^+$ which represents the \emph{free travel time of the edge} (i.e., the minimum travel time needed to travel through the road at maximum allowed speed).
This means that agent $a_i$ setting off at time $t$ from node $v_o$ and heading to node $v_d$ through the edge $(v_o,v_d)$ will reach node $v_d$ at time $t+ \tau(v_o,v_d)$, and will occupy edge $(v_o,v_d)$ only in the time interval $[t, t+\tau(v_o,v_d)]$. Unless stated otherwise, we assume that edges $(u,v)$ and $(v,u)$ are \emph{symmetrical}: for all $(u,v),(v,u)\in E$ $c(u,v)=c(v,u)$, $w(u,v)=w(v,u)$ and $\tau (u,v)= \tau (v,u)$. 

As in \cite{Nesterov2003}, we assume that if the flow of traffic through an edge does not exceed its capacity, then no congestion occurs and the traveling time equals the free travel time.
Initially, at time $t=0$, agents reside on a (publicly known) set $O\subseteq V$ of nodes\footnote{Restricting origins/destinations of journeys to road junctions is without loss of generality since fictitious nodes that serve the sole purpose of acting as starting/ending point of a journey can always be created by edge splitting operations.} of the graph, $O_i$ being the initial location of agent $a_i$.
%Let $\mathbf{O} : A \rightarrow V$ be a public function that assigns agents to their \emph{initial location}.
Each agent $a_i\in A$ wants to reach an intended destination $D_i\in V$, which is the agent's private information and will be referred to in the remainder as her \emph{type}.

Agents submit (or \emph{bid}) a destination to an \emph{allocation mechanism}, which then assigns each agent a path in order to optimize a certain objective function. % the \emph{social cost}, i.e. the sum of all the costs incurred by the agents, subject to the capacity constraint of the edges of the network.
More formally, let $\mathcal{P}$ be the set of all possible simple paths between any two nodes in $G$.
Let $\D=\left(D_1,\ldots,D_n \right)\in V^n$ be a \emph{vector of declarations} (also referred to as \emph{bids}) by the agents and $\DI$ be the vector of declarations of all agents but $a_i$.
%The vector of declarations by all the agents expect $a_i$ will be denoted as $\DI=\{D_1,\ldots,D_{i-1},D_{i+1}, \ldots, D_n\}$.
A \emph{mechanism} $M^{G, O} : V^n \rightarrow \mathcal{P}^n$ maps a vector of declarations to \emph{feasible paths} (i.e., not exceeding the capacity of the edges at any given time) on $G$, given the initial locations $O$ of the agents.
We write $M(\D)$ instead of $M^{G,O}(\D)$ when $G$ and $O$ can be deduced from the context.
The path associated to agent $a_i$ is denoted as $M_i(\D)$.

A traffic assignment $S=M(\D)$ induces a \emph{flow over time}\footnote{Sometimes also referred to as \emph{dynamic flow} in the literature. We prefer the term \emph{flow over} time as the adjective \emph{dynamic} has often been used in many algorithmic settings to refer to problems where the input data 
arrive online or change over time. We assume that all the agents are present at time $t=0$ and the network is cleared after the last agent reaches their destination.} $f_{S} : E\times \mathcal{T} \rightarrow \mathbb{N^+}$, where $\mathcal{T}$ is a suitable discretization of time w.r.t. the transit times of the edges of $G$ (for simplicity we will assume that $\mathcal{T} = \{0,1,\ldots,T\}$, where $T$ is a time horizon sufficient for the network to clear. Thus, $f_{S}(u,v;t) = |\{a_i \in A | (u,v)\in S_i \}|$ is the number of agents that are assigned a path that contains edge $(u,v)$ at time $t\in \mathcal{T}$. Feasibility constraints imply that $f_{S}(u,v;t)\leq c(u,v)$ for all $t\in\mathcal{T}$.
%\begin{definition}

In the remainder, without loss of generality, we will study the problem on the \emph{time-expanded network} \cite{FordFulkerson1,FordFulkersonBook} of $G$ and consider the \emph{static} flow through it (i.e, the transit of an agent over and edge is instantaneous).
A time-expanded network is a properly constructed directed graph with cost and capacity functions on the edges just like $G$, but no transit time
(i.e. travel time is instantaneous through all the edges). For completeness, we give the definition of time expanded networks in the Appendix.
This is without loss of generality from the point of view of SP, Pareto-optimality, non-bossines and approximation guarantee since
it is well known (see
\cite{FordFulkerson1,FordFulkersonBook}) that a flow over time is equivalent to a static flow on the corresponding time-expanded
network.

Let $f^{-i}_S : E \rightarrow \mathbb{N}$ be the flow induced by traffic assignment $S$ generated by agents $A\setminus \{a_i \}$, formally for all $e \in E$, $f^{-i}_{S}(e) = |\{ a_j \in A: e \in S_j, j \neq i\}|$.
The \emph{residual graph} $G_{f}^{-i}$ is a graph such that: (\emph{i}) $G_{f}^{-i}$ has the same nodes and edges as $G$; (\emph{ii}) each edge $e\in E$ of $G_{f}^{-i}$ has capacity $c(e)-{f}^{-i}_S(e)$.
%\end{definition}
For any two nodes $u,v\in V$, let $\mathcal{P}_{u,v}$ denote the set of simple paths in $G$ connecting $u$ to $v$.
Furthermore, for all traffic assignments $S=M(\D)$ and all agents $a_i$, let $\mathcal{P}_{u,v}^i (S) = \{P\in \mathcal{P}_{u,v} | \forall e \in P,  c(e)>{f}^{-i}_{S}(e)\}$.
Informally, $\mathcal{P}_{u,v}^{i}(S)$ is the set of paths connecting $u$ and $v$ that have spare capacity from the perspective of agent $a_i$ (i.e., they can be used by agent $a_i$) when the other agents implement $S$.
Then, the set of reactions available to agent $a_i$ having type $D_i$ at allocation $S$ is defined as $R_i (S) = \mathcal{P}^i_{O_i,D_i} (S)$.

Agents are not constrained to follow their assigned path but can choose a different one, subject to capacity constraints\footnote{We do not prevent agents from using edges other than the ones belonging to their assigned paths, as doing so would result in a waste of public resources (i.e., road capacity). To avoid congestion, though, we assume that agents not following their assigned route can be disincentivized from using an edge that, according to the scheduled traffic, is filled to capacity. This can be easily implemented in a smart traffic control system through the use of traffic cameras that check cars' number plates.}.
To model this, as per \cite{Nissim_2012}, we assume that, after the mechanism computes a traffic allocation, the agents can \emph{react} by choosing an action from a set $R_i \subseteq \mathcal{P}$.
Hence, the actual \emph{cost function} of an agent depends on: (\emph{i}) her true type $D_i$; (\emph{ii}) the allocation $S$ chosen by the mechanism on input the bids reported by the agents; and (\emph{iii}) the reactions chosen by the agents.

We can now formally define the cost function of an agent. Given an allocation $S'=M(D'_i,\DI)$, the cost of an agent of type $D_i$ with respect to $S'$ is defined as:
$cost_i(S',D_i) = \min_{P\in R_i(S')}w(P)$
where $w(P)=\sum_{(u,v)\in P}w(u,v)$ denotes the cost of $P$.
%In the presence of strategic agents, in order to have a theoretical guarantee that agents will not try to manipulate the system  to get a better allocation for themselves we will rely on \emph{strategyproof} mechanisms.
We assume that agents are risk-neutral.
In what follows, we define a set of desiderata for our allocation mechanism, namely: (\emph{i}) strategyproofness, (\emph{ii}) Pareto optimality and (\emph{iii}) non-bossiness.

A deterministic mechanism $M$ is \emph{strategyproof} (SP for short) if, for all agents $a_i$, for all declarations $D_i$ and $D'_i$ and all declarations of the other agents $\DI$, agent $a_i$ cannot decrease her cost by misreporting her true type, namely:
\begin{equation}
 cost_i(M(\D),D_i)\leq cost_i(M(D'_i,\DI),D_i) \label{eq:sp}
\end{equation}
A randomized mechanism is \emph{strategyproof in expectation} if \eqref{eq:sp} holds in expectation (i.e., over the random choices of the mechanism).
A randomized mechanism is \emph{universally strategyproof} if agents cannot gain by
lying regardless of the random choices made by the mechanism, i.e., the output of the mechanism is a
distribution over strategyproof deterministic allocations.

The \emph{social cost} of an allocation $S$ is defined as $SC(S,\D) = \sum_{a_i\in A}cost_i(S,D_i)$.
%We are interested in mechanisms that return allocations that minimize the social cost.
A mechanism $OPT$ is \emph{optimal} for TAP if $OPT(\D) \in \arg\min_{S\in \mathcal{P}^n} SC(S,\D)$ for all $\D$.
%Since the allocation problem for TAP is NP-hard to solve optimally (it is equivalent to the min-cost multi-commodity integral flow problem), and since optimality (as we will show later) is incompatible with strategyproofness, we will have to content ourselves with approximate solutions.
A mechanism $M$ is an $\alpha$--approximation (w.r.t the optimal social cost) with $\alpha\in \mathbb{R}$, $\alpha \geq 1$, being referred to as the \emph{approximation ratio} of $M$, if, for all $\D$, $SC(M(\D),\D)\leq \alpha\cdot SC(OPT(\D),\D)$.

A traffic allocation $S\in \mathcal{P}^n$ is \emph{Pareto optimal} if there exists no other feasible traffic allocation $S'$ such that $cost_j(S',D_j)\leq cost_j(S,D_j)$ for all $a_j$, and $cost_k(S',D_k)<cost_k(S,D_k)$ for some $a_k$.
Pareto optimal allocations are of particular interest in our scenario, because, as proven in Theorem \ref{thm:pareto_opt}, they are a min-cost response in the available reactions $R_i(S)$ of an agent. This gives us a theoretical guarantee that agents will actually implement Pareto optimal solutions returned by the mechanism.

\begin{theorem}\label{thm:pareto_opt}
Let $S=M(\D)$ be a traffic assignment and let $R_i(S)$ be the set of reactions available to $a_i$ at $S$.
If $S$ is Pareto optimal, then $M_i(\D)\in \arg\min_{P\in R_i(S)}w(P)$.
\end{theorem}
\iftoggle{proof_thm_1}{
\begin{proof}
Let us suppose by contradiction that there exists a reaction $r'_i\in R_i$ that is feasible and strictly better than $M_i(\D)$, i.e. $c(r'_i)<c(M_i(\D))$.
Then, a route assignment $M'$ such that $M'_j = M_j$ for all $j\neq i$ and $M'_i = r_i$ is still feasible.  Since $cost_j(M')=cost_j(M')$ for all $j\neq i$, and $cost_i(M')<cost_i(M)$, $M$ is not Pareto optimal. 
\end{proof}
}

%\textbf{Definition of the ordinal and cardinal version of TAP. In Section \ref{sec:deterministic_mechs} we give a characterization for these two kinds of mechanisms.}
%We will restrict our attention to \emph{non-bossy} mechanisms. 

Finally, a mechanism $M$ is \emph{non-bossy} if $M_i(\D) = M_i(D'_i,\DI)$ implies that $M_j(\D) = M_j(D'_i,\DI)$, for all $a_i,a_j\in N$ and all $\D$ and $D'_i$. In other words, non-bossyness excludes (arguably undesirable) mechanisms that allow one agent to change the allocation of other agents without changing her own too.
In the remainder of this paper, we focus on strategyproof mechanisms for TAP that approximately achieve the optimal social cost. In particular, we are interested in mechanisms that are also Pareto-optimal and non-bossy.

%In the remainder, we will also study slight variations of the TAP problem that will help us get some insights into the properties of the ``original'' TAP problem.
%In accordance with the literature, we name \emph{cardinal TAP} and \emph{ordinal TAP} \cite{Filos-RatsikasF014}.
%Cardinal TAP takes in input for each agent $i$ a cost function $d_i:\mathcal{P} \rightarrow \mathbb{R}$ such that $d(P)$ is the cost incurred by agent $i$ if path $P\in \mathcal{P}$ is assigned to her.
%Ordinal TAP takes in input an ordering $\succeq_i$ over $\mathcal{P}$ such that $P_{k_1} \succeq_i P_{k_2}$, for $P_{k_1},P_{k_2} \in \mathcal{P}$ if agent $i$ prefers path $P_{k_1}$ over path $P_{k_2}$.
%It is worth noting that our TAP problem can be reduced to both cardinal TAP and ordinal TAP, whereas the opposite direction does not hold.
%We will use these two problems in the negative, by proving impossibility results on cardinal TAP and ordinal TAP.
%Although these negative results do not automatically extend to the TAP problem (i.e. the set of SP mechanism can be bigger for TAP than it is for cardinal and ordinal TAP), they give us a strong hint about the dichotomy between strategyproofness and approximability.

%\input{results.tex}
% !TEX root = ijcai18.tex
\section{Deterministic Mechanisms}\label{sec:deterministic_mechs}
In this section, we discuss deterministic mechanisms for TAP.
In particular, we first provide a lower bound on the approximation ratio of SP deterministic mechanisms.

\begin{theorem}\label{thm:deterministic_LB}
There is no $\alpha$-approximate deterministic SP mechanism for the traffic assignment problem with $\alpha<3 - \varepsilon$, for any $\varepsilon>0$.
\end{theorem}
\iftoggle{proof_thm_1_sketch}{
\begin{proof}
%%% sketch of the proof
\emph{(Sketch)}
Consider the graph depicted in Figure \ref{fig:LB_instance_2}, where the labels on the edges represent their capacity (red) and length (black). There are two agents $a_1$ and $a_2$ at node $A$, whose intended destinations are, respectively, $D$ and $G$.
The length of path $(B,C)$ is $K = \min\{2, \frac{10-4\varepsilon}{\varepsilon}\}$ and the length of path $(F,E)$ is $K-1$. 
\begin{figure}[h]
\centering
\includegraphics[width=.55\linewidth]{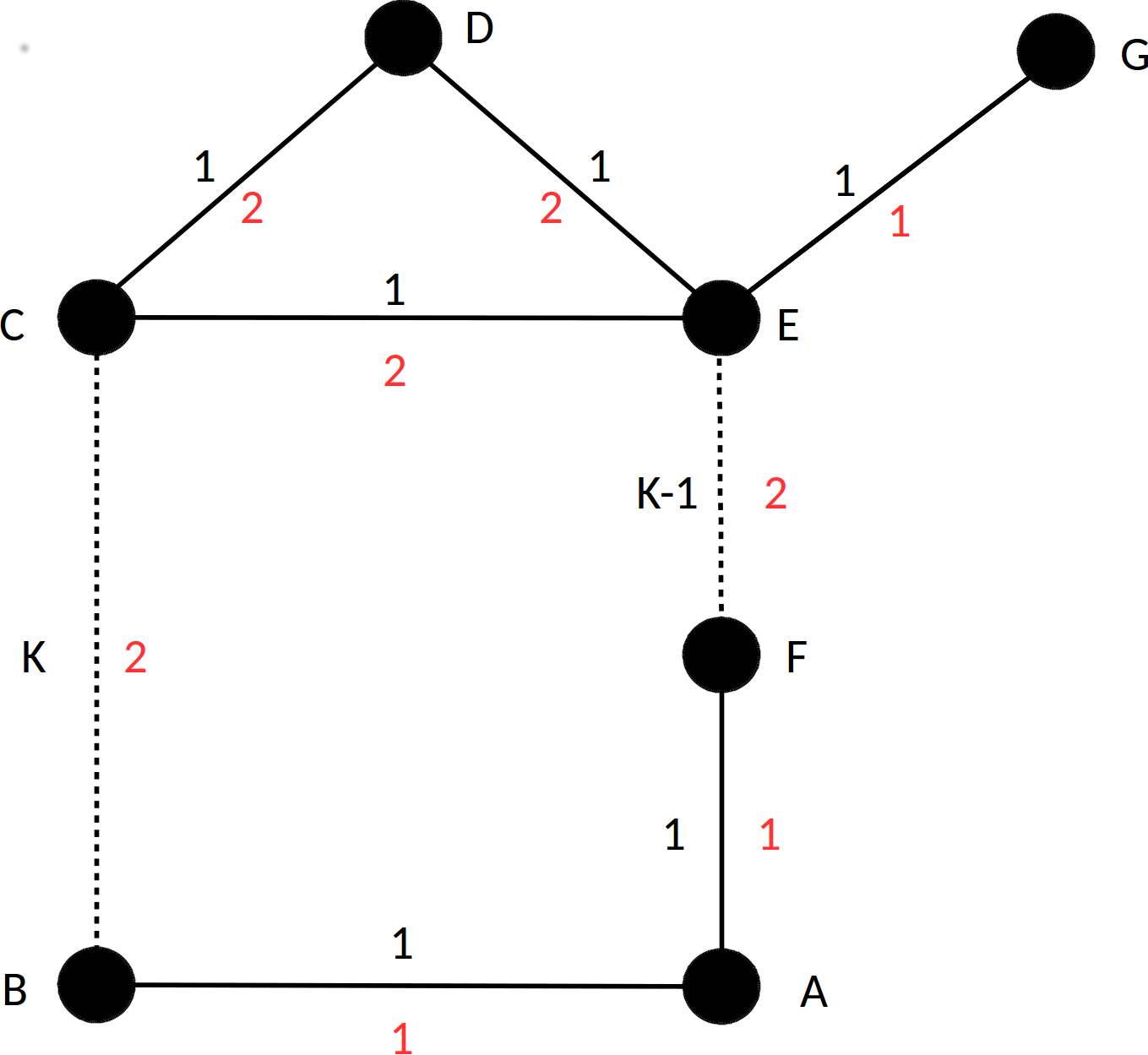}
\caption{Lower bound instance}\label{fig:LB_instance_2}
\end{figure}
%Let us consider a generic $\alpha$-approximate mechanism $M$.
%Assume by contradiction that $M$ is strategyproof and $\alpha$-approximate with $\alpha<3-\varepsilon$. %, where $\varepsilon = \frac{10}{k+4}$. 
The instance has two Pareto optimal solutions, depending on which player is allocated the edge $(A,F)$. The optimal allocation $\mathbf{P}^*=OPT(\D)$ is $P^*_1 = (A,B,C,D)$ and $P^*_2 = (A,F,E,G)$, $cost_1(\mathbf{P}^*,D) = K+2$ and $cost_2(\mathbf{P}^*,G) = K+1$, and $SC(\mathbf{P}^*,\D)=2K+3$. The second best solution is $P_1 = (A,F,E,D)$ and $P_2 = (A,B,C,E,G)$. %We note that $cost_1(P_1,D_1)=K+1$ and $cost_2(P_2,D_2)=K+3$, for a social cost of $SC(M(\D),\D) = 2K+4$. %; observe that this solution has an approximation ratio lower than $3-\varepsilon}$ when $K\geq \frac{31+\sqrt{1761}}{8}$.
%Regardless of the solution $M$ returns on this instance, %there is another instance where to maintain SP, 
$M$ achieves an approximation not better than $3-\varepsilon$, a contradiction.
If $M$ returns the optimal allocation, agent $a_1$ declares $D'_1=F$ instead of her true type. By SP $M$ cannot allocate the edge $(A, F)$ to $a_1$ ($a_1$ can use path $(A,F,E,D)$ and reach her true destination).% In fact, assume for the sake of contradiction that $M(D_1', D_2)$ allocates $(A,F)$ to $a_1$.
%Then $a_2$ is allocated path $(A,B,C,E,G)$ and $a_1$ can use path $(A,F,E,D)$ and reach her true destination, thus having: 
% $$cost_1(M(D'_1,D_2),D_1)=K+1<cost_1(M(\D),D_1)=K+2.$$ 
Therefore, $M(D_1', D_2)$ must return $P'_1 = (A,B,C,E,F)$ and $P'_2 = (A,F,E,G)$, with $SC((P'_1,P'_2), (D_1', D_2)) = 3K+2$, whilst the optimum for this instance is $OPT_1(D'_1,D_2) = (A,F)$ and $OPT_2(D'_1,D_2) =(A,B,C,E,G)$. The social cost of the optimum is then $SC(OPT(D'_1,D_2),(D'_1,D_2)) = K+4$.
Therefore $M$ has an approximation ratio higher than $3-\varepsilon$.
If $M(\D)$ returns $M_1(\D) = P_1$ and $M_2(\D) = P_2$, then agent $a_2$ reports $D'_2=F$ instead of her true type.
The optimal allocation is $OPT_1(D_1,D_2') = (A,B,C,D)$ and $OPT_2(D_1,D_2') = (A,F)$, and $SC(OPT(D'_2,D_1),(D'_2,D_1)) = K+3$
%As before, this allocation is not strategyproof as:
%\begin{eqnarray*}
%cost_2(OPT_2(D_1, D'_2),D_2)= K+1 \\ 
%< cost_2(M_2(D_1,D'_2),D_2)=K+3 
%\end{eqnarray*}
(i.e., agent $a_2$ can use the route $(A,F,E,G)$ to reach her true destination).
In this case the best (in terms of approximation ratio) SP allocation is $P'_1=(A,F,E,D)$ and $P'_2=(A,B,C,E,F)$, with a cost of $SC((P'_1,P'_2),(D_1,D_2')) = 3K+2$.
This solution has an approximation ratio higher than $3-\varepsilon$.
\end{proof}
}{
\iftoggle{proof_thm_1}{
\begin{proof}
Given $\varepsilon>0$, consider the graph depicted in Figure \ref{fig:LB_instance_2}, where the labels on the edges represent their capacity (red) and length (black). The instance we consider has two agents $A= \{a_1,a_2\}$, both initially located at node $A$, whose intended destination is $D$ and $G$, respectively (namely, $\D=(D,G)$).
The length of the path $(B,C)$ is $K = \min\{2, \frac{10-4\varepsilon}{\varepsilon}\}$ and the length of path $(F,E)$ is $K-1$. 
\begin{figure}[h]
\centering
\includegraphics[width=.55\linewidth]{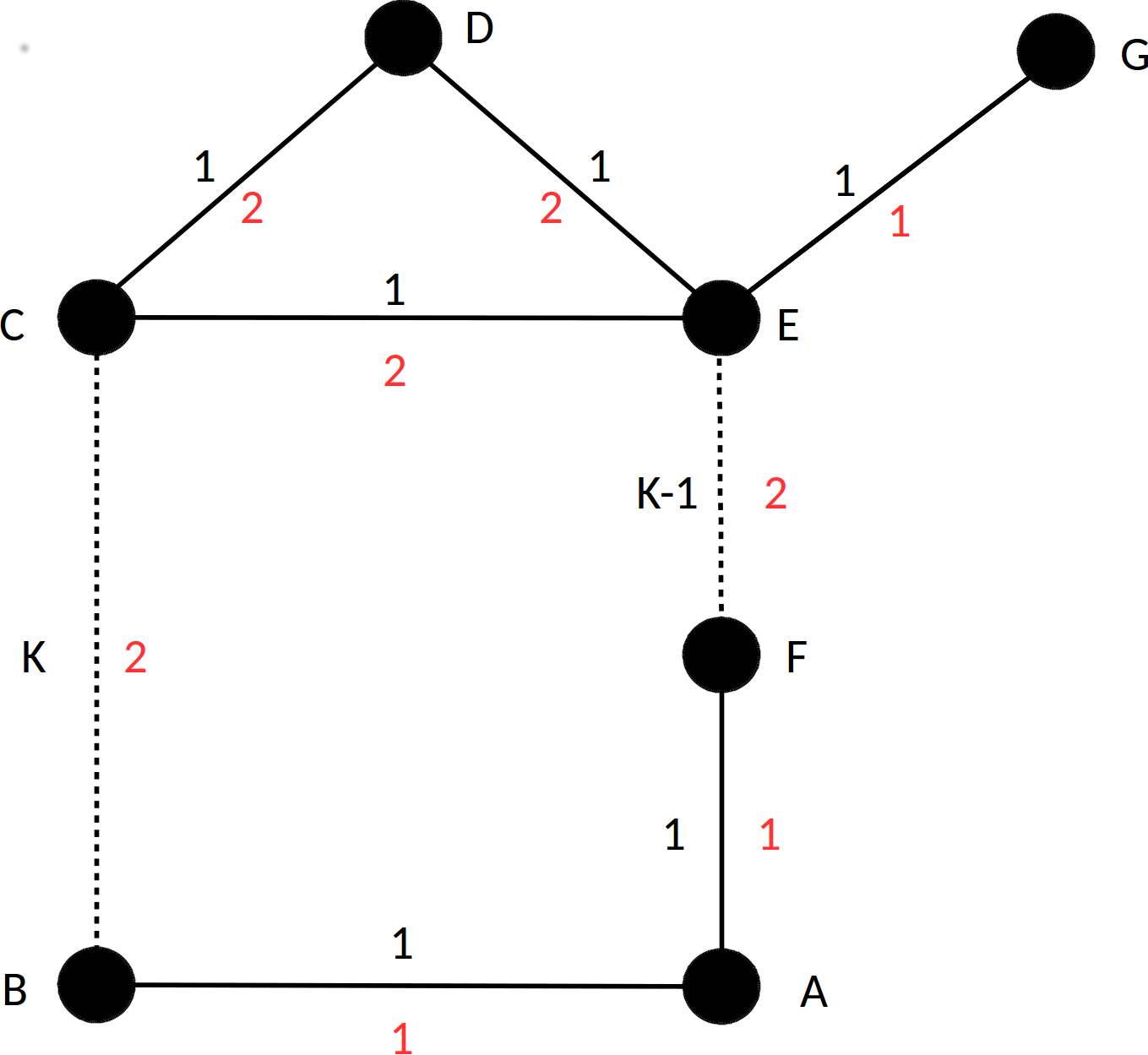}
\caption{Lower bound instance}\label{fig:LB_instance_2}
\end{figure}
Let us consider a generic $\alpha$-approximate mechanism $M$.
Assume by contradiction that $M$ is strategyproof and $\alpha$-approximate with $\alpha<3-\varepsilon$. %, where $\varepsilon = \frac{10}{k+4}$. 
The instance has two Pareto optimal solutions, depending on which player is allocated the edge $(A,F)$ (note that only one agent at a time can use edge $(A,F)$ as its capacity is 1). The optimal allocation $\mathbf{P}^*=OPT(\D)$ is $P^*_1 = (A,B,C,D)$ and $P^*_2 = (A,F,E,G)$, $cost_1(\mathbf{P}^*,D) = K+2$ and $cost_2(\mathbf{P}^*,G) = K+1$, and $SC(\mathbf{P}^*,\D)=2K+3$. The second best solution is $P_1 = (A,F,E,D)$ and $P_2 = (A,B,C,E,G)$. We note that $cost_1(P_1,D_1)=K+1$ and $cost_2(P_2,D_2)=K+3$, for a social cost of $SC(M(\D),\D) = 2K+4$. %; observe that this solution has an approximation ratio lower than $3-\varepsilon}$ when $K\geq \frac{31+\sqrt{1761}}{8}$.
We are going to prove that, regardless of the solution $M$ returns on this instance, there is another instance where to maintain SP, $M$ achieves an approximation not better than $3-\varepsilon$, a contradiction.
Let us assume first that $M$ returns the optimal allocation.
If agent $a_1$ declares $D'_1=F$ instead of her true type, by SP, $M$ cannot allocate the edge $(A, F)$ to $a_1$. In fact, assume for the sake of contradiction that $M(D_1', D_2)$ allocates $(A,F)$ to $a_1$.
Then $a_2$ is allocated path $(A,B,C,E,G)$ and $a_1$ can use path $(A,F,E,D)$ and reach her true destination, thus having: 
 $$cost_1(M(D'_1,D_2),D_1)=K+1<cost_1(M(\D),D_1)=K+2.$$ 
Therefore, $M(D_1', D_2)$ must return $P'_1 = (A,B,C,E,F)$ and $P'_2 = (A,F,E,G)$, with $SC((P'_1,P'_2), (D_1', D_2)) = 3K+2$, whilst the optimum for this instance is $OPT_1(D'_1,D_2) = (A,F)$ and $OPT_2(D'_1,D_2) =(A,B,C,E,G)$. The social cost of the optimum is then $SC(OPT(D'_1,D_2),(D'_1,D_2)) = K+4$.
Therefore $M$ has an approximation ratio higher than $3-\varepsilon$.
Let us now suppose that $M(\D)$ returns $M_1(\D) = P_1$ and $M_2(\D) = P_2$.
In this case, consider the case that agent $a_2$ reports $D'_2=F$ instead of her true type.
The optimal allocation is $OPT_1(D_1,D_2') = (A,B,C,D)$ and $OPT_2(D_1,D_2') = (A,F)$, and $SC(OPT(D'_2,D_1),(D'_2,D_1)) = K+3$.
As before, this allocation is not strategyproof as:
\begin{eqnarray*}
cost_2(OPT_2(D_1, D'_2),D_2)= K+1 \\ 
< cost_2(M_2(D_1,D'_2),D_2)=K+3 
\end{eqnarray*}
(i.e., agent $a_2$ can use the route $(A,F,E,G)$ to reach her true destination).
As above, one can easily check that in this case the best (in terms of approximation ratio) strategyproof allocation is $P'_1=(A,F,E,D)$ and $P'_2=(A,B,C,E,F)$, with a cost of $SC((P'_1,P'_2),(D_1,D_2')) = 3K+2$.
This solution has an approximation ratio higher than $3-\varepsilon$.
\end{proof}
}
}

The above theorem  implies the following corollary:
%at the optimal allocation cannot be SP. That is:
%, as its approximation ratio is $1$, which is smaller than 3. That is, we have:
%From the above theorem me infer the following:

\begin{corollary}
The optimal allocation is not strategyproof for TAP.
\end{corollary}

%The main implication of this corollary is that solutions that only focus on optimizing the social cost, and thus, ignoring the strategic behaviour of the participating agents, are not strategyproof. 
These impossibility results suggest that in order to achieve strategyproofness we have to give up on optimality.
This naturally leads to asking to what extent can we approximate the optimal social welfare while satisfy the desired properties. 
As a first step to answer this question, we examine the well-known Serial Dictatorship mechanism that is deterministic and notoriously satisfies our three desiderata (i.e., strategyproofness, Pareto optimality and non-bossiness).
\begin{definition}
Mechanism Serial Dictatorship (SD), given an ordering $a_1\prec, \ldots, \prec a_n$ of the agents, allocates paths to agents in $n$ stages such that at stage $i$ agent $a_i$ is allocated her minimum cost path in the residual graph $G_f^{-\{a_1,\ldots,a_{i-1}\}}$.
\end{definition}

The following theorem proves that SD is indeed feasible under some mild conditions:
\begin{theorem}
If $G$ is $K$-edge-connected\footnote{A graph is $K$-edge-connected if it remains connected when strictly fewer than $K$ edges are removed.}, mechanism SD is feasible for $K$ agents.
\end{theorem}
\iftoggle{proof_thm_3}{
\begin{proof}
%Given a fixed ordering $a_1\prec \ldots \prec a_k$ over the agents, Serial Dictator works by assigning to each agent $a_1$ is most preferred path, and, for all $j>1$ assigning agent $a_j$ his most preferred path among the ones remaining after agents $\{a_1,\ldots,a_{j-1}\}$.
If the graph is $K$-edge connected, the allocation returned by the Serial Dictator will always be feasible, (i.e. paths assigned to different agents will not overlap and there is always an assignable path for each agents).
This follows from the fact that in a $K$-edge-connected graph there are at least $K$ edge disjoint paths between any pair of nodes.
\end{proof}
}
%It is worth highlighting that $SD$ is notoriously SP and Pareto optimal, hence the agents are incentivized to implement the allocation returned by the algorithm.

Next we provide an upper bound on the approximation ratio of SD, and thus, on its worst case performance.
In order to prove our result, we make the following assumption:
\begin{definition}
The \emph{deviation on capacious path assumption (DoCP)} assumes that whenever the SD mechanism allocates  to an agent a path that is different from the one that the optimal mechanism would allocate, the assigned path has sufficient capacity to potentially be allocated to all the remaining agents. 
%Furthermore, the agent cannot react to this allocation by deviating on less capacious paths.
\end{definition}
To better understand this assumption, consider the following example.
With reference to Figure \ref{fig:capacious_paths}, let $a_i$ be an agent and $P^*_i$ be the path she is assigned in the optimal allocation (i.e., $OPT_i= P^*_i$).
If agent $a_i$ is not assigned $P^*_i$ by SD, there must be an agent $a_j$, where $j\prec i$ in the ordering used by SD, such that: (\emph{i}) $SD_j= P_j\neq OPT_j$ and (\emph{ii}) $P_j\cap P^*_i \neq \emptyset$ and (\emph{iii}) at least one edge of $P^*_i$ is saturated after $a_j$ is assigned $P_j$.
In such a situation, we say that agent $a_i$ is blocked by agent $a_j$.
Let $\alpha_i \in P_j\cap P^*_i$ ($\beta_i \in P_j\cap P^*_i$, respectively) be the first (last, respectively) node of $P^*_i$ in $P_j$. 
The DoCP assumption postulates that if $a_j$ blocks $a_i$, then the \emph{alternative path of blocked agent $a_i$ through blocking agent $a_j$} $\Gamma_{i}^{j} = (O_i,\alpha _i,O_j,D_j,\beta_i,D_i)$ has at least capacity $n-|\{a_k \in A | a_j \prec a_k\}|$ in the residual graph $G_f^{-\{a_1,\ldots
,a_{j}\}}$. That is, all agents yet to be assigned by SD after $a_j$ can be accommodated on this path.
We note that, by construction, if agent $a_i$ is blocked by agent $a_j$ then path $\Gamma^{j}_{i}$ always exists, although unless we assume DoCP, it might not have spare capacity to be assigned to agent $a_i$.
\begin{figure}
\centering
\includegraphics[width=.40\linewidth]{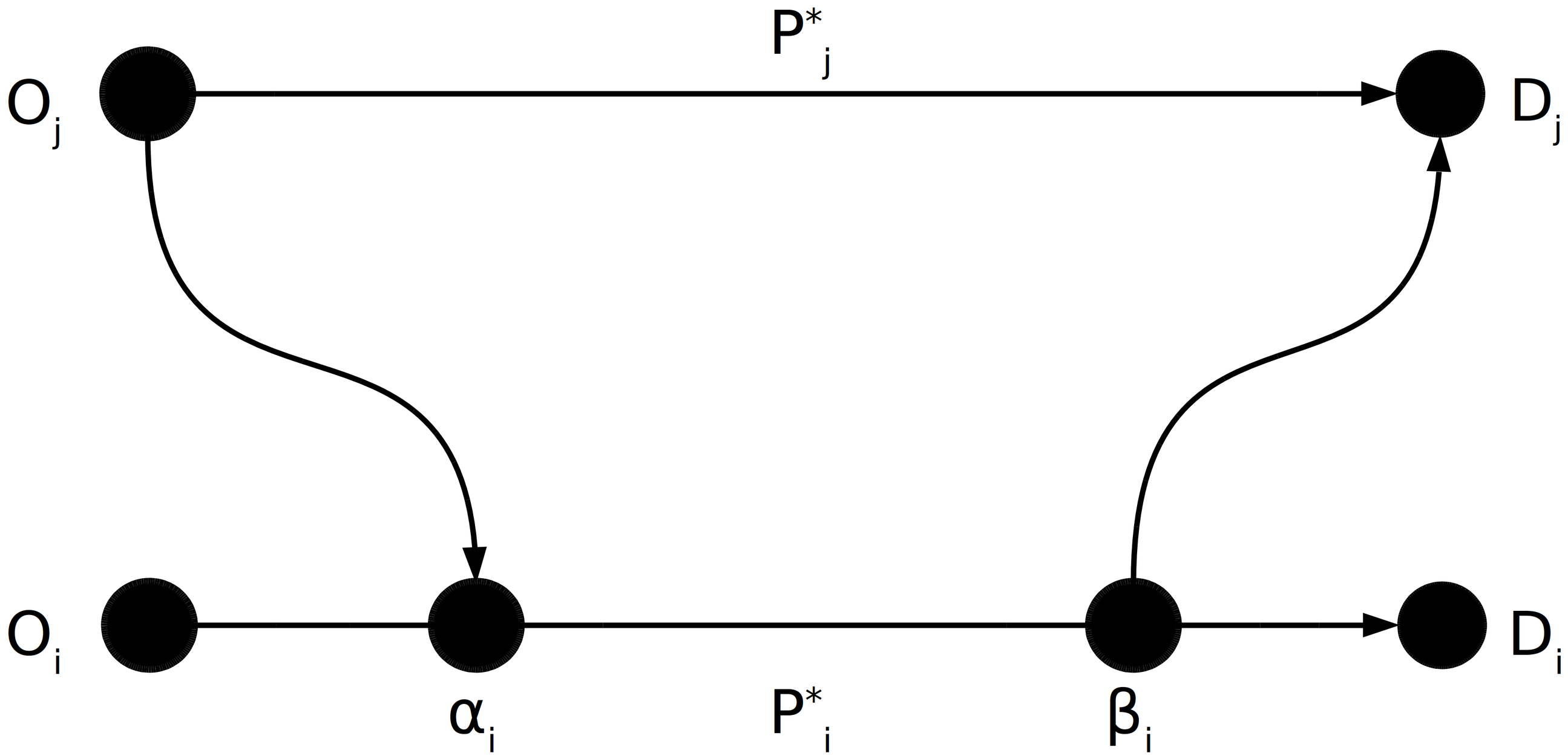}
\caption{Deviation on capacious paths}\label{fig:capacious_paths}
\end{figure}
It is not difficult to see that if we relax the DoCP assumption, then the approximation ratio of SD is not bounded by any function of the number of agents on certain pathological TAP instances. 
%We omit a formal proof of this simple fact due to space limitations.

\begin{theorem}\label{thm:apx_SD}
Under the DoCP assumption, SD is at most $(2^n-1)$-approximate.
\end{theorem}
\iftoggle{proof_thm_4_sketch}{
\begin{proof}[Proof sketch]
We prove the claim by induction on the number of players. Let $OPT_i$ denote the cost and solution (with a slight abuse of notation) of the optimal allocation that only considers bids of agents $j \leq i$. Similarly, let $SD_i$ denote the cost and solution of $SD$ on input all the bids of agents $j \leq i$.
Base of the induction ($i=1$): trivially $OPT_1=SD_1$.
Now assume that the claim is true for $i-1$ and, for $j \leq i$, let $P^*_j$ ($P_j$, respectively) be the path assigned to agent $j$ by $OPT_i$ ($SD_i$, respectively). For a path $P$, we let $w(P)$ denote the cost of the path in the given graph G.
We want to prove that under the DoCP assumption, the following holds:
\begin{equation}\label{eq:paths}
w(P_i) \leq OPT_i + SD_{i-1}.
\end{equation}
If $P^*_i=P_i$ then we are done.
Therefore, we can assume that $P^*_i \neq P_i$.
This means that the paths $P_j$ allocated to agents $j < i$ by $SD_i$ saturate some of the edges of $P^*_i$.
Now, for at least one of these agents, say $\bar{j}$, $P^*_{\bar{j}} \neq P_{\bar{j}}$ for otherwise also in $OPT_i$ path $P^*_i$ would be unavailable to $i$. But then $w(P_i) \leq w(\Gamma^{\bar{j}}_i)$, $\Gamma^{\bar{j}}_i$ being the path that connects $O_i$ to $D_i$ through $O_{\bar{j}}$, as per the definition of DoCP.
Note that, under the DoCP assumption, $\Gamma^{\bar{j}}_i$ is always feasible.
Since $\Gamma^{\bar{j}}_i$ uses only edges in $OPT_i \cup SD_{i-1}$ (i.e. $P^*_i$ and $P^*_j$ are in $OPT_i$, paths $(O_i,\alpha_i)$ and $(\beta_i,D_j)$ belong to $SD_{i-1}$), \eqref{eq:paths} is proven.
We finally observe that \eqref{eq:paths} and the inductive hypothesis yield:
\begin{align*}
SD_i & = SD_{i-1} + w(\Gamma^{\bar{j}}_i)  \leq 2 SD_{i-1} + OPT_i \\
& \leq  2((2^{i-1}-1) OPT_{i-1}) + OPT_i \leq (2^i-1) OPT_i.
\end{align*}
\end{proof}
}{
\iftoggle{proof_thm_4}{
\begin{proof}
We are going to prove the claim by induction on the number of players. Specifically, let $OPT_i$ denote the cost and solution (with a slight abuse of notation) of the optimal allocation that only considers bids of agents $j \leq i$. Similarly, $SD_i$ denotes the cost and solution of $SD$ on input all the bids of agents $j \leq i$.
For the base of the induction with $i=1$ it is clear that $OPT_1=SD_1$.
Now assume that the claim is true for $i-1$ and, for $j \leq i$, let $P^*_j$ ($P_j$, respectively) be the path assigned to agent $j$ by $OPT_i$ ($SD_i$, respectively). For a path $P$, we let $w(P)$ denote the cost of the path in the given graph G.
We want to prove that under the DoCP assumption, the following holds:
\begin{equation}\label{eq:paths}
w(P_i) \leq OPT_i + SD_{i-1}.
\end{equation}
Let us begin by observing that if $P^*_i=P_i$ then we are done.
Therefore, we can assume that $P^*_i \neq P_i$.
This means that the paths $P_j$ allocated to agents $j < i$ by $SD_i$ saturate some of the edges of $P^*_i$.
Now, for at least one of these agents, say $\bar{j}$, $P^*_{\bar{j}} \neq P_{\bar{j}}$ for otherwise also in $OPT_i$ path $P^*_i$ would be unavailable to $i$. But then $w(P_i) \leq w(\Gamma^{\bar{j}}_i)$, $\Gamma^{\bar{j}}_i$ being the path that connects $O_i$ to $D_i$ through $O_{\bar{j}}$, as per the definition of DoCP.
Note that, under the DoCP assumption, $\Gamma^{\bar{j}}_i$ is always feasible.
Since $\Gamma^{\bar{j}}_i$ uses only edges in $OPT_i \cup SD_{i-1}$ (i.e. $P^*_i$ and $P^*_j$ are in $OPT_i$, paths $(O_i,\alpha_i)$ and $(\beta_i,D_j)$ belong to $SD_{i-1}$), \eqref{eq:paths} is proven.
We can then conclude the proof, by observing that \eqref{eq:paths}, along with the inductive hypothesis, yield:
\begin{align*}
SD_i = SD_{i-1} + w(\Gamma^{\bar{j}}_i) & \leq 2 SD_{i-1} + OPT_i \\
& \leq  2((2^{i-1}-1) OPT_{i-1}) + OPT_i \\ & \leq (2^i-1) OPT_i. 
\end{align*}
\end{proof}
}
}

As the $(2^n-1)$-approximation ratio can be prohibitively large for large $n$, we ask ourselves whether we can further improve this upper bound. Unfortunately, the following theorem answers this question in the negative.
\begin{theorem}\label{thm:tightness_SD}
Under the DoCP assumption, the bound of Theorem \ref{thm:apx_SD} is tight.
\end{theorem}
\iftoggle{proof_thm_5}{
\begin{proof}
Let us consider the instance in Figure \ref{fig:sd_tight_instance}, where there are $n$ nodes $v_1,\ldots, v_n$ and $n$ agents $A=\{a_1,\ldots, a_n\}$ such that agent $a_i$ is initially located at node $v_i$. All agents want to reach the same destination $D$.
Each link has capacity $1$.
Agent $a_1$ has two paths to her destination $D$: one direct path that costs $1+\varepsilon$ (where $\varepsilon\ll 1$ is a small constant) and a path costing $1$ that goes through the node agent $a_2$ is initially located on.
Each agent $a_i$, for $i=2,\dots,n-1$ has two paths: one direct path that costs $\varepsilon$ and a path costing $2^{i-1}$ that goes through the node  agent $a_{i+1}$ is initially located on.
Agent $a_n$ has two direct paths, costing $\varepsilon$ and $2^{n-1}$ respectively.
The optimal traffic assignment assigns agent $a_1$ to the path that costs $1+\varepsilon$ and the other agents to the path costing $\varepsilon$, and has a cost of $1+\varepsilon n$.
Let us consider ordering $a_1\prec a_2 \prec \ldots \prec a_n$.
On this ordering, mechanism SD assigns agent $a_1$ the path costing $1$, and to each agent $a_i$, for $i=2,\ldots, a_n$ the path costing $2^{i-1}$ for a total cost of $\sum_{i=0}^{n-1} 2^i = 2^n-1$.
For $\varepsilon$ close to $0$, the approximation ratio of SD on the instance depicted in Figure \ref{fig:sd_tight_instance} is hence close to $2^n-1$.
\begin{figure}[h]
\centering
\includegraphics[width=.6\linewidth]{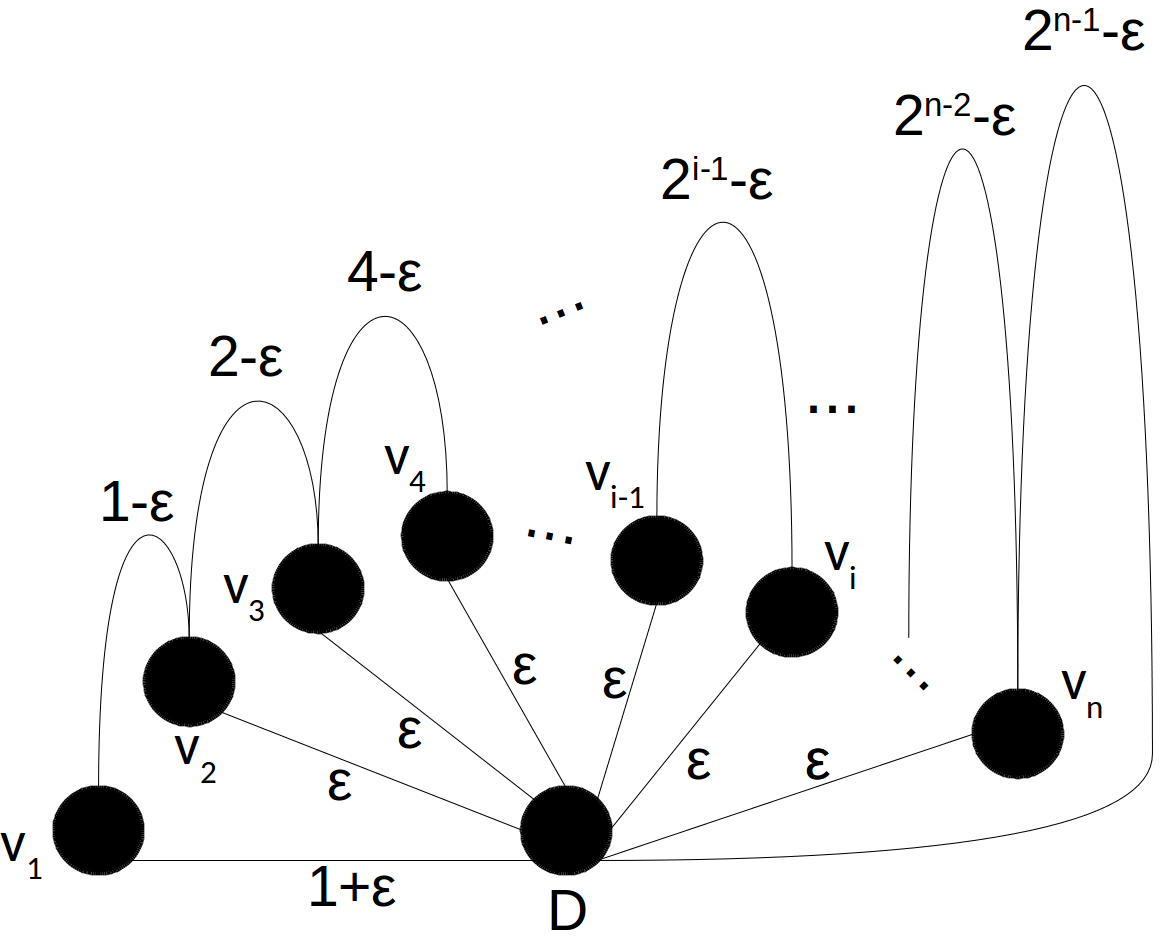}
\caption{Tight instance for SD}\label{fig:sd_tight_instance}
\end{figure}
\end{proof}
}

We now provide a characterization of SP, Pareto-optimal, and non-bossy mechanisms for a subset of instances of TAP, named \TAPStar\, and we prove that the family of all mechanisms satisfying the above properties is comprised by a generalization of SD, namely \emph{Bi-polar Serial Dictatorship} (BSD).
Such a characterization extends naturally to TAP instances.
\TAPStar{} is subset of instances of TAP having a peculiar structure: (\emph{i}) every agent has the same source node $O$; (\emph{ii}) $O$ has outgoing edges with unitary capacity and no ingoing edges, let $E_O = \{(O,v_1),\ldots, (O,v_m)\}$ denote the set of outgoing edges of $O$; and (\emph{iii}) the set of possible destinations that the agents can declare is restricted to a given subset $\mathcal{D}\subset V$.

\begin{definition}
Given an ordering of the agents $\{i_1,i_2\}\prec i_3\prec \ldots\prec i_n$ and a bipartition $\{X_1,X_2\}$ of the set of alternatives  $X$ (i.e., paths in the case of TAP) such that $X_1\cap X_2 = \emptyset$ and $X_1 \cup X_2 = X$, a BSD mechanism executes SD with ordering $i_2\prec i_1 \prec \ldots \prec i_n$ if $\min_{x\in X} cost_1(x)  = \min_{x\in X} cost_2(x) = x\in X_2$; otherwise SD with ordering $i_1\prec i_2\prec \ldots \prec i_n$ is executed.\end{definition}

\begin{theorem}\label{thm:reduction}
A traffic allocation mechanism for \TAPStar is Pareto-optimal, SP and non-bossy if and only if it is a Bi-polar Serially Dictatorial Rule.
\end{theorem}
\iftoggle{proof_thm_6_sketch}{
\begin{proof}[Proof sketch]
We reduce an instance of the problem of \emph{assigning indivisible objects} with general ordinal preferences \cite{DBLP:journals/jet/BogomolnaiaDE05} (AIO for short) to \TAPStar{}.
In an instance of AIO, a set of objects $X = \{x_1,\ldots ,x_m\}$ has to be assigned to a set of agents $A = \{a_1,\ldots ,a_n\}$, such that every agent receives at most one object and no agent is left without an object if there are objects still available.
Agents have ordinal general preferences $\succeq _i$, where $x \succeq_i y$ for $x,y\in X$ means that agent $i$ (weakly) prefers object $x$ to object $y$.
From an instance of AIO, we build an instance of \TAPStar{} as follows. 
\TAPStar{} has the same set of agents $A$ as AIO.
Graph $G$ of \TAPStar{} has a node $O$ such that $O_i=O$ for all $a_i \in A$.
For every object $x_j\in X$ we construct in $G$ a node $v_j$ and an edge $(O,v_j)$ such that $c(O,v_j)=1$ and $w(O,v_j) = \varepsilon$ for $0<\varepsilon\ll 1$.
Let $\Psi$ be the set of all possible preference relations over $X$. 
We construct $|\Psi|$ destination nodes $D_k$, one for each preference relation $\succeq \in \Psi$ and for each $k\in{1,\ldots, |\Psi|}$.
For each $j \in \{1,\ldots, m \}$ we add an edge $(v_j,D_k)$ having capacity 1 and weight $w(v_j,D_k)$ equal to the \emph{ranking} of $x_j$ according to $\succeq$.
%Figure \ref{fig:reduction_example} gives an example of the reduction applied to an AIO game with $A=\{a_1,a_2,a_3\}$, $X=\{x_1,x_2,x_3\}$ and $\Psi$ being the set of all possible \emph{linear orderings} over $X$. %The labels on the edges of the graph of Figure \ref{fig:reduction_example} represent the costs of the edges, whereas all capacities are set to 1.
%By construction, the following hold: (\emph{i}) any path allocation on $G$ must include all the edges $(O,v_j)$; (\emph{ii}) any edge $(O,v_j)$ is used by at most one path; and (\emph{iii}) only one agent can be assigned any given edge $(O,v_j)$ due to the capacity constraint.
We can now transform an instance of the so-constructed \TAPStar{} problem to an instance of the AIO problem, and vice versa.
%Indeed, let $P_i$ be the path assigned to agent $a_i$ in the \TAPStar{} problem. If $P_i$ contains edge $(O,v_j)$ we allocate object $x_j$ to agent $a_i$ in the AIO instance, and vice versa from an allocation for the AOI problem to an allocation for the \TAPStar{} problem.
%by simply constructing a node $v_j$ and an edge $(S,v_j)\in E_S$ in $G$ for each object $x_j\in X$ for each edge.
In \cite{DBLP:journals/jet/BogomolnaiaDE05} it is proved that BSD is the only Pareto optimal, SP and non-bossy mechanism for AIO.
This characterization transfers to \TAPStar{} due to the reduction sketched above.
\end{proof}

}{
\iftoggle{proof_thm_6}{
\begin{proof}
We will reduce an instance of the problem of \emph{assigning indivisible objects} with general ordinal preferences \cite{DBLP:journals/jet/BogomolnaiaDE05} (AIO for short) to \TAPStar{}.
%In the ordinal version of TAP, agents submit a preference relation over the set of paths, instead of a single node representing their destination.
%We will show that this is without loss of generality from the point of view of characterizing cardinal mechanisms.
An instance of AIO is composed of a set of objects $X = \{x_1,\ldots ,x_m\}$ that have to be assigned to a set of agents $A = \{a_1,\ldots ,a_n\}$, such that every agent receives at most one object and no agent is left without an object if there are objects still available.
Agents have ordinal general preferences $\succeq _i$, where $x \succeq_i y$ for $x,y\in X$ means that agent $i$ (weakly) prefers object $x$ to object $y$.
From an instance of AIO, we can build an instance of \TAPStar{} as follows. 
\TAPStar{} has the same set of agents $A$ as AIO.
Graph $G$ of \TAPStar{} has a node $O$ such that $O_i=O$ for all $a_i \in A$.
For every object $x_j\in X$ we construct in $G$ a node $v_j$ and an edge $(O,v_j)$ such that $c(O,v_j)=1$ and $w(O,v_j) = \varepsilon$ for $0<\varepsilon\ll 1$.
Let $\Psi$ be the set of all possible preference relations over $X$. 
We construct\footnote{We note that, although $|\Psi|$ can be exponential in $m$, it is always finite. We remark that graphs of exponential size are not an issue here since the characterization we are proving in this theorem does not rely on computational efficiency.} $|\Psi|$ destination nodes $D_k$, one for each preference relation $\succeq \in \Psi$ and for each $k\in{1,\ldots, |\Psi|}$.
For each $j \in \{1,\ldots, m \}$ we add an edge $(v_j,D_k)$ having capacity 1 and weight $w(v_j,D_k)$ equal to the \emph{ranking}\footnote{
The alternatives in $X$ can be partitioned in subsets ${X_1,\ldots, X_\ell,\ldots}$ such that any two elements $x_1,x_2\in X_\ell$ are indifferent according to $\succeq$ and, for any $x_1\in X_\ell$ and $x_2 \in X_{\ell+1}$, $x_1$ is strictly preferred to $x_2$ according to $\succeq$. Then $\ell$ is the ranking of $x\in X_\ell$.} of $x_j$ according to $\succeq$.
Figure \ref{fig:reduction_example} gives an example of the reduction applied to an AIO game with $A=\{a_1,a_2,a_3\}$, $X=\{x_1,x_2,x_3\}$ and $\Psi$ being the set of all possible \emph{linear orderings} over $X$. The labels on the edges of the graph of Figure \ref{fig:reduction_example} represent the costs of the edges, whereas all capacities are set to 1.
%\begin{table}[]
%\centering
%\caption{Mapping elements of $\Psi$ to destination nodes of $G$}
%\label{my-label}
%\begin{tabular}{|l|l|}
%\hline
%$D_1$ & $x_1\prec x_2\prec x_3$ \\ \hline
%$D_2$ & $x_1\prec x_3\prec x_2$ \\ \hline
%$D_3$ & $x_2\prec x_1\prec x_3$ \\ \hline
%$D_4$ & $x_2\prec x_3\prec x_1$ \\ \hline
%$D_5$ & $x_3\prec x_1\prec x_2$ \\ \hline
%$D_6$ & $x_3\prec x_2\prec x_1$ \\ \hline
%\end{tabular}
%\end{table} 

By construction, the following hold: (\emph{i}) any path allocation on $G$ must include all the edges $(O,v_j)$; (\emph{ii}) any edge $(O,v_j)$ is used by at most one path; and (\emph{iii}) only one agent can be assigned any given edge $(O,v_j)$ due to the capacity constraint.
We can now easily transform an instance of the so-constructed \TAPStar{} problem to an instance of the AIO problem, and vice versa.
%Indeed, let $P_i$ be the path assigned to agent $a_i$ in the \TAPStar{} problem. If $P_i$ contains edge $(O,v_j)$ we allocate object $x_j$ to agent $a_i$ in the AIO instance, and vice versa from an allocation for the AOI problem to an allocation for the \TAPStar{} problem.
%by simply constructing a node $v_j$ and an edge $(S,v_j)\in E_S$ in $G$ for each object $x_j\in X$ for each edge.
In \cite{DBLP:journals/jet/BogomolnaiaDE05} it is proved that BSD is the only Pareto optimal, SP and non-bossy mechanism for AIO.
This characterization trivially transfers to \TAPStar{} due to the reduction sketched above.
Indeed, let us suppose that there exists an SP, Pareto optimal and non-bossy mechanism for TAP.
Such mechanism would be SP, Pareto optimal and non-bossy for the AIO instance as well.
%It is easy to see that the characterization obtained for the ordinal version of cardinal \TAPStar{} too.
%Indeed, let us assume for the sake of contradiction that there exists a cardinal mechanism $M$ for the TAP problem that is not BSD but it is SP, Pareto optimal and non-bossy.
%We can easily build an ordinal mechanism $M'$ from $M$ as follows.
%Mechanism $M'$  as follows. 
%For every agent $i$, $M'$ computes makes up a cost function for $i$ that is consistent with $\succeq_i$.
%Path $cost_i(P) = 1$ for all $P \in \mathcal{P}^i_1 = max_{\succeq_i} \{{\mathcal{P}}\}$ and $cost_i(P) = 2$ for all $P\in \mathcal{P}^i_2=max_{\succeq_i} \{\mathcal{P}\setminus \mathcal{P}^i_1\}$ and $cost_i(P) = k$ if $P\in \mathcal{P}^i_2=max_{\succeq_i} \{\mathcal{P}\setminus \mathcal{P}^i_{k-1}\}$ .
%Subsequently, $M'$ feeds the cardinal preferences to mechanism $M$ and returns $M$'s output on such input.
%\begin{figure}[t]
%\centering
%\includegraphics[width=\linewidth]{figures/reduction_bigger.png}
%\caption{Reduction example}\label{fig:reduction_example}
%\end{figure}
\end{proof}
}
}
%This theorem states that for \TAPStar, the only SP, neutral, and non-bossy mechanism is BSD. 
Next, we investigate the performance of BSD and show that it does not asymptotically perform better than SD.
In particular, we state that:
\begin{lemma}
BSD cannot achieve an approximation ratio lower than $\Omega(2^n)$ for TAP.
\end{lemma}
\iftoggle{proof_lemma_1}{
\begin{proof}
We are going to show an instance of TAP where BSD has an approximation ratio of $\Omega(2^n)$.
Let us take the instance of Figure \ref{fig:sd_tight_instance} and let us consider the ordering $\{a_1,a_2\}\prec a_3\prec \ldots \prec a_n$.
Let us consider $X_1 = \{v_2,D\}$ and $X_2 = E\setminus X_1$.
The so-defined BSD mechanism, on input the instance of figure \ref{fig:sd_tight_instance} would always execute SD with ordering $a_1 \prec a_2 \prec \ldots, a_n$. We know from Theorem \ref{thm:tightness_SD} that under this ordering the approximation ratio of SD is $\Omega(2^n)$.
\end{proof}
}

\section{Randomized Mechanisms}\label{sec:randomized_mechs}
Given the undesirable approximation guarantees of deterministic mechanisms, we now turn to randomization. Randomized mechanisms can often be interpreted as fractional mechanisms for the deterministic solutions, under mild conditions.
We start by proving the following inapproximability lower bound:

\begin{theorem}\label{thm:randomizedLB}
There is no $\alpha$-approximate universally truthful randomized mechanism for the traffic assignment problem with $\alpha<11/10$.
\end{theorem}
\iftoggle{proof_thm_7}{
\begin{proof}
Our approach is based on Yao's minimax principle \cite{Yao}. In our context, this principle states that the approximation ratio of the best universally truthful randomized mechanism is equal to the approximation ratio of the best deterministic truthful mechanism under a worst-case input distribution. Accordingly, we exhibit a probability distribution over input instances for which any deterministic truthful mechanism cannot attain an approximation guarantee better than $11/10$.
The two instances are taken from the proof of Theorem \ref{thm:deterministic_LB}, where we set $K=2$. Specifically, we consider the instance in Figure \ref{fig:LB_instance_2}, that we name $I$, and the very same instance where agent $a_1$ reports $F$; we call this instance $I'$. We consider a probability distribution over $I$ and $I'$ that returns $I$ with probability $\lambda$ and $I'$ with the remaining probability $1-\lambda$, where $\lambda=2/3$.  The expected value of the optimum will then be $(\lambda+1)K+4-\lambda=20/3$.
Let $M$ be a SP deterministic mechanism. From the arguments in the proof of the theorem above, we know that $M$ must assign the edge $(A,F)$ to the same agent in both instances $I$ and $I'$. If $M$ allocates $(A,F)$ to agent $a_1$ in both the instances then its expected social cost will be $(\lambda+1)K+4=22/3$ for an approximation ratio of $11/10$. If instead $M$ allocates $(A,F)$ to agent $a_2$ in both the instances then the expected social cost of the mechanism will be $(3-\lambda)K+\lambda+2=22/3$; the approximation ratio of $M$ would then be $11/10$.
\end{proof}
}
%
%The above theorem allow us to infer the following corollary:
%\begin{corollary}
%The optimal allocation for TAP is not strategyproof
%\end{corollary}

In the remainder of this section, we study the randomized version of SD for TAP, which is universally strategyproof, (ex-post) Pareto optimal and non-bossy.

\begin{definition}
The Randomized Serial Dictatorship (RSD) mechanism computes uniformly at random an ordering $\sigma$ over the agents and returns the output of SD over ordering $\sigma$.
\end{definition}

The following results gives a tight bound on the approximation ratio of RSD. 
\begin{theorem}
\label{theorem:RSD_upperbound}
Under the DoCP assumption, RSD is at most $n$-approximate.% for TAP.
\end{theorem}

\iftoggle{proof_thm_8}{
\begin{proof}[Proof sketch]
We are going to prove the claim by induction on the number of agents. As above, let $OPT_i$ denote the cost of the optimal solution with paths assigned only to agents $a_j$, with $j \leq i$. With a slight abuse of notation we also let $OPT_i$ denote the solution itself. Similarly, $RSD_i$ denotes the expected cost of $RSD$ on input all the bids of agents $a_j$, $j \leq i$.
For the base of the induction with $i=1$, it is clear that $RSD_1$ is the optimal solution.
Now assume that the claim is true for $i-1$ and consider an instance with $i$ agents. 
	%For a player $k \leq i$, we denote with $\pi_k$ the path that minimizes the cost of agent $j$ (i.e., the path that $SD$ would assign to $j$ if she was the first to choose). In a slight abuse of notation, we let $SC_{OPT}(L, c')$ ($SC_{RSD}(L, c')$) be the (expected) social cost of the optimum (RSD) on any instance with path requests of agents in $L$ on the graph $G$ where $e \in E$ has capacity $c'(e) \leq c(e)$. 
	%
Let $I_{-k}(P)$, $P$ being a path from $O_k$ to $D_k$, be the instance of the problem without agent $a_k$ and with the capacity of the directed edges in $P$ diminished by one (i.e., as if the path $P$ were used by $a_k$).
Note that by the DoCP assumption, one of the agents $a_j$, with $j \neq k$, is guaranteed to be able to use the edges of $P$ in the opposite direction than $a_k$. 
%{\bf FALSE: $k$ is not in the instance!}
We now let $OPT_{-k,P}$ and $RSD_{-k,P}$ be the cost of the optimum and expected cost of RSD on $I_{-k}(P)$, respectively. Moreover, let $\pi_j$ be the path minimizing the cost of agent $a_j$ (i.e., the path that $SD$ would assign to $a_j$ if she was the first to choose). We then have 
{\scriptsize
	\begin{align*}
	RSD_i &= \frac1i \sum_{k=1}^i \left(\rule{0ex}{3ex}w(\pi_k) + RSD_{-k,\pi_k}\right)
	 \leq \frac1i \sum_{k=1}^i \left(\rule{0ex}{3ex}w(\pi_k) + (i-1) OPT_{-k, \pi_k}\right)\\
	& \leq \frac1i \sum_{k=1}^i w(\pi_k) + \frac1i \sum_{k=1}^i \left(\rule{0ex}{3ex}(i-1) (OPT_i + w(\pi_k))\right)
	 \leq \frac1i OPT_i + (i-1) OPT_i + \frac{i-1}{i}OPT_i \\
	 & = i \cdot OPT_i	
	\end{align*}
}
\noindent
where the first equality follows from the definition of RSD, i.e., with probability $1/i$ each agent $k$ will have the first choice. As for the inequalities, we note that the first follows from the inductive hypothesis whilst the last from the observation that $OPT_i \geq \sum_{k=1}^i w(\pi_k)$. We are left with the second inequality. That is, we prove that under the DoCP
	$
	OPT_{-k, \pi_k} \leq OPT_i + w(\pi_k).
	$
If $OPT_i$ allocates $\pi_k$ to agent $a_k$ then we are done.
Otherwise, let $P_k$ be the path that $a_k$ gets in $OPT_i$ and note that the paths $P_j$ allocated to agents $a_j$ $j \neq k$ by $OPT_i$ saturates some of the edges of $P_k$; let $a_{\bar{j}}$ be one of these agents.
Consider now the solution $S$ to $I_{-k}(\pi_k)$ where all agents but $a_{\bar{j}}$ are allocated the same path as in $OPT_i$ and agent $a_{\bar{j}}$ is given, instead of $P_{\bar{j}}$, the alternative path $\Gamma_{\bar{j}}^{k}$ through agent $a_k$.
Observe that $\Gamma_{\bar{j}}^{k}$ uses the same directed edges of $P_{\bar{j}}$ and $P_k$ and the edges of $\pi_k$ in opposite direction and, as observed above, under the DoCP assumption, is a feasible path for $a_{\bar{j}}$ and $S$ a feasible solution to $I_k(\pi_k)$, whose social cost is denoted $SC(S)$. But then:
	\begin{align*}
	OPT_{-k, \pi_k} & \leq SC(S) = OPT_i - w(P_j) - w(P_k) + w(P)
	 \leq OPT_i + w(\pi_k)
	\end{align*}    
	where the last inequality follows from the fact that the edges in $P \setminus (P_k \cup P_j)$ are a subset of the edges in $\pi_k$.
\end{proof}
}

\begin{theorem}
\label{theorem:RSD_lowerbound}
The approximation ratio of RSD is $\Omega(n)$.
\end{theorem}

\noindent
This means that by allowing randomness in the allocation mechanism, we can improve the exponential approximation ratio of the deterministic case to a linear one. 
%The proof of Theorem~\ref{theorem:RSD_upperbound} is sketched as follows:

\iftoggle{proof_thm_9_sketch}{
\begin{proof}[Proof sketch of Theorem~\ref{theorem:RSD_lowerbound}]
The proof uses the same construction as the instance of Figure \ref{fig:sd_tight_instance}, with $k<n$ nodes.
One agent is initially located at node $v_1$, whereas $1+2\cdot3^{i-1}$ agents are initially located at node $v_i$, for $i=2,\ldots, k-1$ .
With a little abuse of notation, let $|v_i|$ denote the number of agents initially located at node $v_i$, and let $n_i = \sum_{\ell=0}^i |v_\ell|$.
Edges $(v_1,D)$ and $(v_1,v_2)$ have capacity $1$, whereas edges $(v_i,D)$ and $(v_i,v_{i+1})$ have capacity $1+2\cdot3^{i-1}$ for $i>1$.
Let $a_1 \prec\ldots \prec a_n$ be an ordering over the agents.
We will be interested in orderings that possess the \emph{chain of levels} property, namely for all $i=1,2,\ldots, k-1$ at least one agent located at node $i$ appears after all agents of levels $0,1,\ldots i-1$.
The property of a chain of levels ordering with respect to the instance of Figure \ref{fig:sd_tight_instance} is that it forces at least one agent located at node $v_i$, for all $i = 1,\ldots, k-1$ to use the path $P = (v_i,v_{i+1},D)$, at a cost of $2^{i-1}$ for the agents, and an overall social cost of $\sum_{i=1}^{k-1}2^{i-1}=2^k-1 > 2^{k-1}$.

We argue that the probability that a chain of levels ordering is chosen by RSD is $\Pi_{i=1}^{k-1} \left(1-\frac{n_{i-1}}{n_i}\right)$.
Indeed, we can look at the process of randomly generating an ordering as follows. First an ordering for the agents located at each node is uniformly generated at random.
Then orderings of agents of consecutive nodes are merged together in lexicographic order.
In particular, we start merging the orderings of nodes $v_1$ and $v_2$.
There are $\binom{|v_2|+|v_1|}{v_1} = \binom{n_2}{n_1}$ such orderings, whereas there are $\binom{n_2-1}{n_1}$ orderings where one agent located at node $v_2$ follows all the agents located at node $v_1$. The partial ordering obtained so far is randomly merged with the ordering of agents at node $v_3$ and the procedure continues until the partial ordering is complete. When merging agents at node $v_i$ with the current partial ordering, we note that there are $\binom{n_i}{n_{i-1}}$ possible orderings, and $\binom{n_i-1}{n_{i-1}}$ where for all $\ell = 1,\ldots, i$, one agent located at node $v_\ell$ follows all the agents located at node $v_{\ell-1}$ in the ordering (i.e., fix one agent from node $v_\ell$ in the last position and compute all possible orderings of the other agents).
Hence, the probability of one agent at node $v_\ell$ appearing after all agents at node $v_{\ell-1}$ is $\binom{n_i-1}{n_{i-1}}/\binom{n_i}{n_{i-1}} = \left(1-\frac{n_{i-1}}{n_i}\right)$.
Since the random orderings generated at each stage are independent, the probability that for all $i = 1, 2, \ldots, k-1$ at least one agent at node $v_i$ appears after all agents at node $v_{i-1}$ in a random ordering is $\Pi_{i=1}^{k-1}\left(1-\frac{n_{i-1}}{n_i}\right)$.
Hence, the probability that a chain of levels ordering is chosen by RSD for the instance of Figure \ref{fig:sd_tight_instance} is $(2/3)^{k-1}$.

Finally, the expected cost of RSD is at least $(4/3)^{k-1} = n^{\log_3(4/3)}\approx n^{0.262}$. Since the optimal allocation costs $1+\epsilon \cdot n$, the approximation ratio is $\Omega(n)$ for $\epsilon$ close to $0$.
\end{proof}

}{
\iftoggle{proof_thm_9}{
\begin{proof}
The proof uses the same construction as the instance of Figure \ref{fig:sd_tight_instance}, with $k<n$ nodes.
One agent is initially located at node $v_1$, whereas $1+2\cdot3^{i-1}$ agents are initially located at node $v_i$, for $i=2,\ldots, k-1$ .
With a little abuse of notation, let $|v_i|$ denote the number of agents initially located at node $v_i$, and let $n_i = \sum_{\ell=0}^i |v_\ell|$.
Edges $(v_1,D)$ and $(v_1,v_2)$ have capacity $1$, whereas edges $(v_i,D)$ and $(v_i,v_{i+1})$ have capacity $1+2\cdot3^{i-1}$ for $i>1$.
Let $a_1 \prec\ldots \prec a_n$ be an ordering over the agents.
We will be interested in orderings that possess the \emph{chain of levels} property, namely for all $i=1,2,\ldots, k-1$ at least one agent located at node $i$ appears after all agents of levels $0,1,\ldots i-1$.
The property of a chain of levels ordering with respect to the instance of Figure \ref{fig:sd_tight_instance} is that it forces at least one agent located at node $v_i$, for all $i = 1,\ldots, k-1$ to use the path $P = (v_i,v_{i+1},D)$, at a cost of $2^{i-1}$ for the agents, and an overall social cost of $\sum_{i=1}^{k-1}2^{i-1}=2^k-1 > 2^{k-1}$.

We argue that the probability that a chain of levels ordering is chosen by RSD is $\Pi_{i=1}^{k-1} \left(1-\frac{n_{i-1}}{n_i}\right)$.
Indeed, we can look at the process of randomly generating an ordering as follows. First an ordering for the agents located at each node is uniformly generated at random.
Then orderings of agents of consecutive nodes are merged together in lexicographic order.
In particular, we start merging the orderings of nodes $v_1$ and $v_2$.
There are $\binom{|v_2|+|v_1|}{v_1} = \binom{n_2}{n_1}$ such orderings, whereas there are $\binom{n_2-1}{n_1}$ orderings where one agent located at node $v_2$ follows all the agents located at node $v_1$. The partial ordering obtained so far is randomly merged with the ordering of agents at node $v_3$ and the procedure continues until the partial ordering is complete. When merging agents at node $v_i$ with the current partial ordering, we note that there are $\binom{n_i}{n_{i-1}}$ possible orderings, and $\binom{n_i-1}{n_{i-1}}$ where for all $\ell = 1,\ldots, i$, one agent located at node $v_\ell$ follows all the agents located at node $v_{\ell-1}$ in the ordering (i.e., fix one agent from node $v_\ell$ in the last position and compute all possible orderings of the other agents).
Hence, the probability of one agent at node $v_\ell$ appearing after all agents at node $v_{\ell-1}$ is $\binom{n_i-1}{n_{i-1}}/\binom{n_i}{n_{i-1}} = \left(1-\frac{n_{i-1}}{n_i}\right)$.
Since the random orderings generated at each stage are independent, the probability that for all $i = 1, 2, \ldots, k-1$ at least one agent at node $v_i$ appears after all agents at node $v_{i-1}$ in a random ordering is $\Pi_{i=1}^{k-1}\left(1-\frac{n_{i-1}}{n_i}\right)$.
Hence, the probability that a chain of levels ordering is chosen by RSD for the instance of Figure \ref{fig:sd_tight_instance} is $(2/3)^{k-1}$.

Finally, the expected cost of RSD is at least $(4/3)^{k-1} = n^{\log_3(4/3)}\approx n^{0.262}$. Since the optimal allocation costs $1+\epsilon \cdot n$, the approximation ratio is $\Omega(n)$ for $\epsilon$ close to $0$.
\end{proof}
}
}

%This implies that the result in Theorem 9 is tight. That is, RSD cannot achieve better approximation ratio than $n$.
%Now we are going to prove an approximation lower bound for randomized mechanisms on \TAPStar{}
%\begin{theorem}
%No truthful randomized mechanism has approximation ratio lower than $\Omega{(\sqrt{n})}$ for \TAPStar{}.
%\end{theorem}
\iftoggle{proof_thm_10}{
\begin{proof}
We reduce the ordinal AIO problem studied in \cite{Filos-RatsikasF014} to \TAPStar{} through the same reduction as Theorem \ref{thm:reduction}.
Note that AIO problem in \cite{Filos-RatsikasF014} requires that the $|X| = |A|$ but this can be easily accommodated.
In \cite{Filos-RatsikasF014} the authors prove that no truthful in expectation mechanism for the AIO problem can achieve an approximation ratio lower than $\Omega{(\sqrt{n})}$. In virtue of the reduction from ordinal AIO to \TAPStar{}, this results holds for \TAPStar{} as well.
\end{proof}
}

\begin{figure*}%[t!]
\begin{minipage}{0.49\linewidth}
\includegraphics[width=0.49\linewidth]{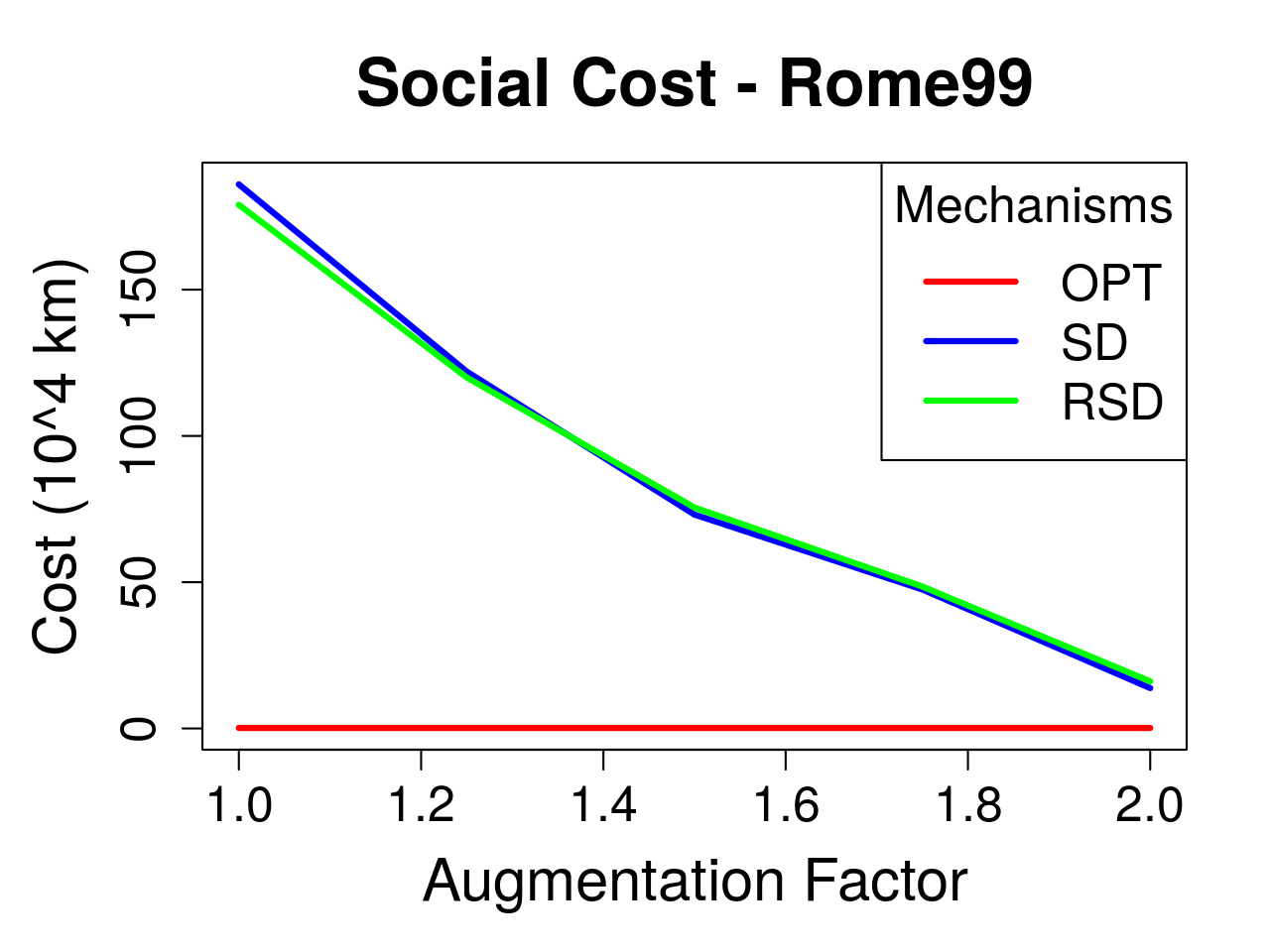}
\includegraphics[width=0.49\linewidth]{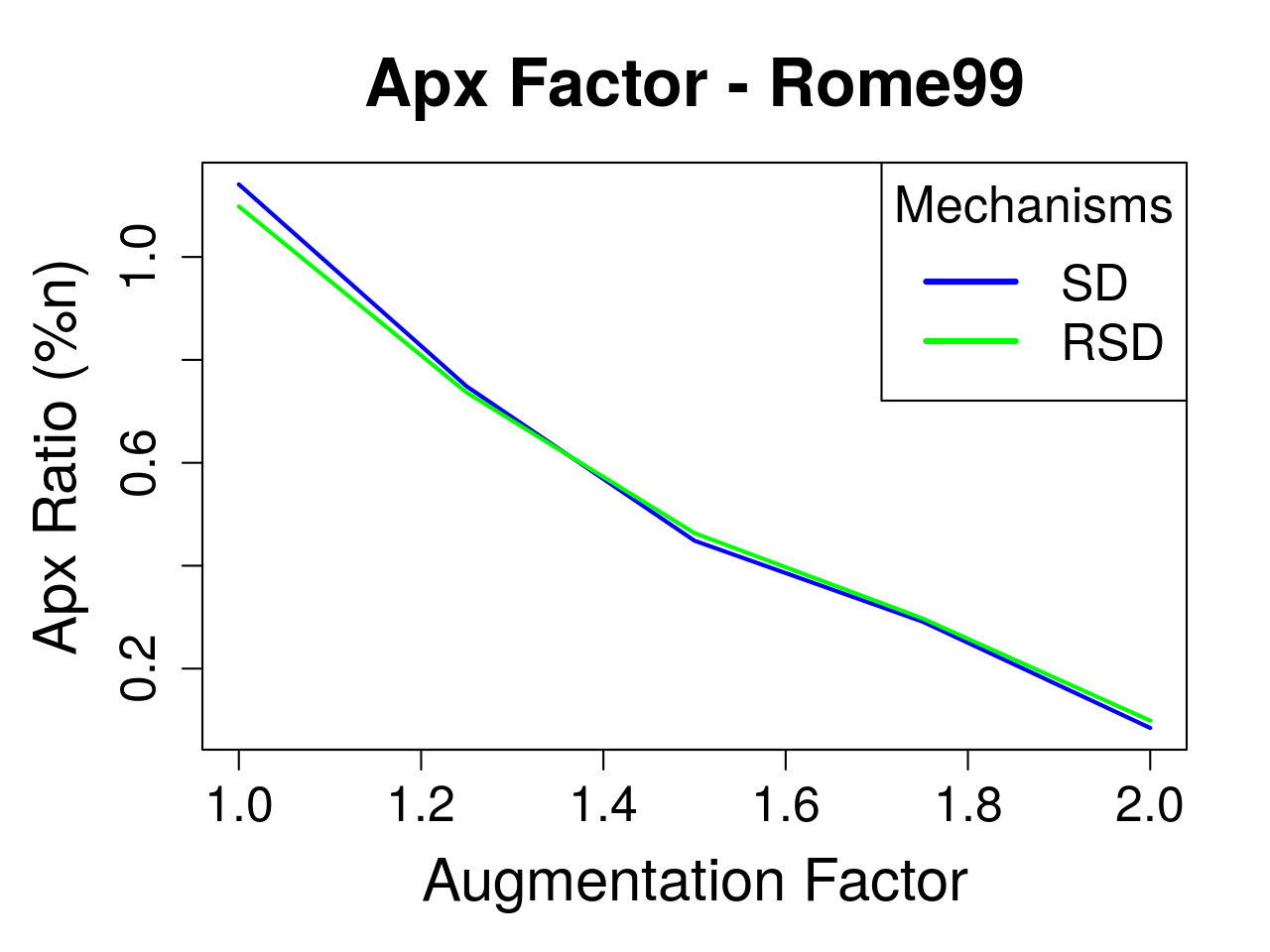}
\caption{Experimental results on Rome99}\label{fig:Rome_results}
\end{minipage}
\begin{minipage}{0.49\linewidth}~~~~~~
\begin{tabular}{l|l|l|l|l|}
\cline{2-5}
                                 & Rome-99 & NY-4000 & \multicolumn{2}{l|}{NY-10000} \\ \hline
\multicolumn{1}{|l|}{$|V|$}        & 3000 & 4000    & \multicolumn{2}{l|}{10000}    \\ \hline
\multicolumn{1}{|l|}{$|E|$}        & 8859 & 10027   & \multicolumn{2}{l|}{312594}   \\ \hline
\multicolumn{1}{|l|}{$\delta^+_{AVG}$} & 2.6  & 2.5     & \multicolumn{2}{l|}{31}       \\ \hline
\multicolumn{1}{|l|}{$c_{AVG}$}   & 27.3 & 20.5    & \multicolumn{2}{l|}{30}       \\ \hline                                                    
\end{tabular}
\caption{Structural characteristics of test graphs}\label{tab:graphs_stats}
\end{minipage}
\end{figure*}

\begin{figure*}[t!]
\begin{minipage}{0.49\linewidth}
\includegraphics[width=0.49\linewidth]{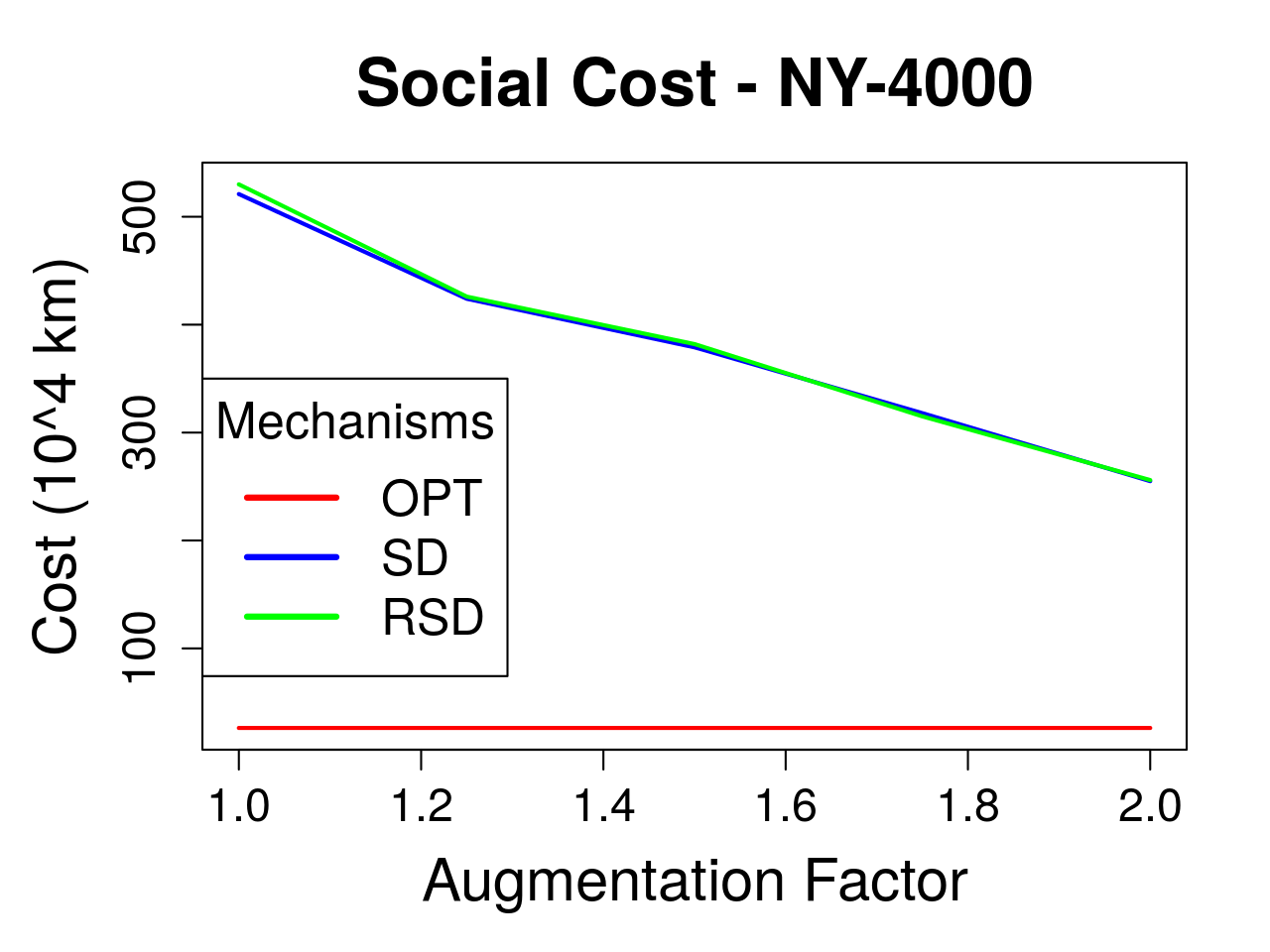}
\includegraphics[width=0.49\linewidth]{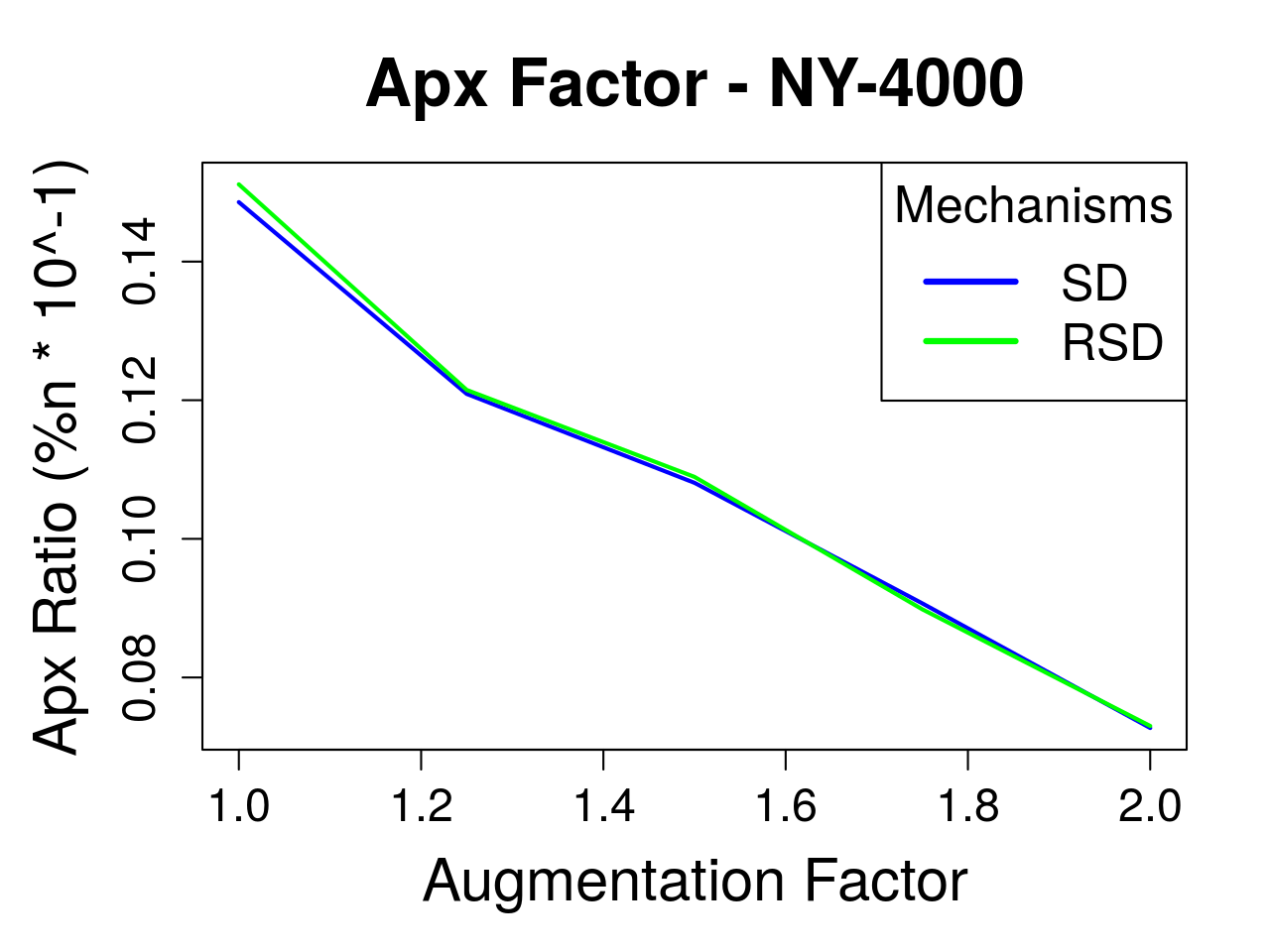}
\caption{Experimental results on NY-4000}\label{fig:NY-4000_results}
\end{minipage}
\begin{minipage}{0.49\linewidth}
\includegraphics[width=0.49\linewidth]{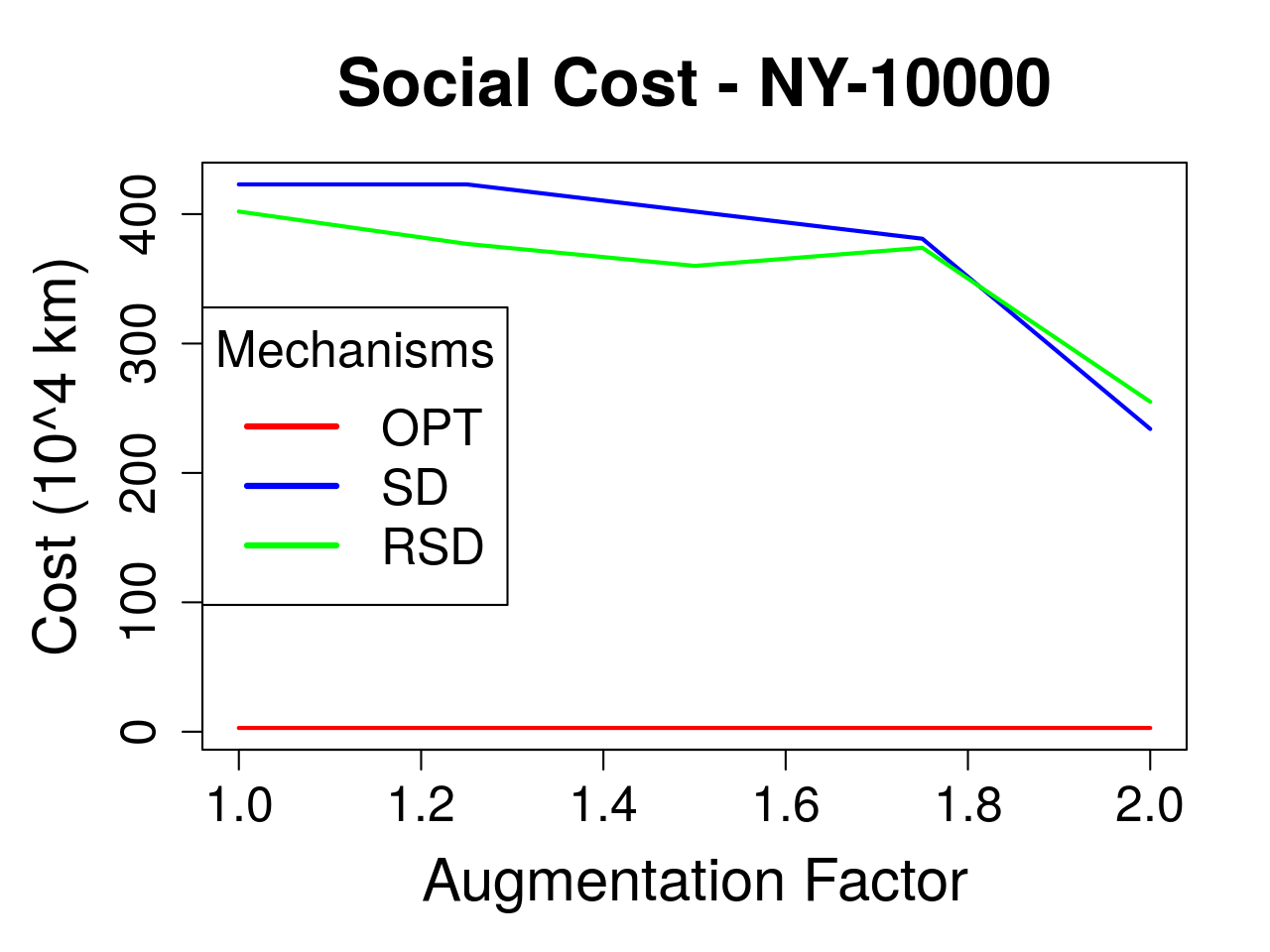}
\includegraphics[width=0.49\linewidth]{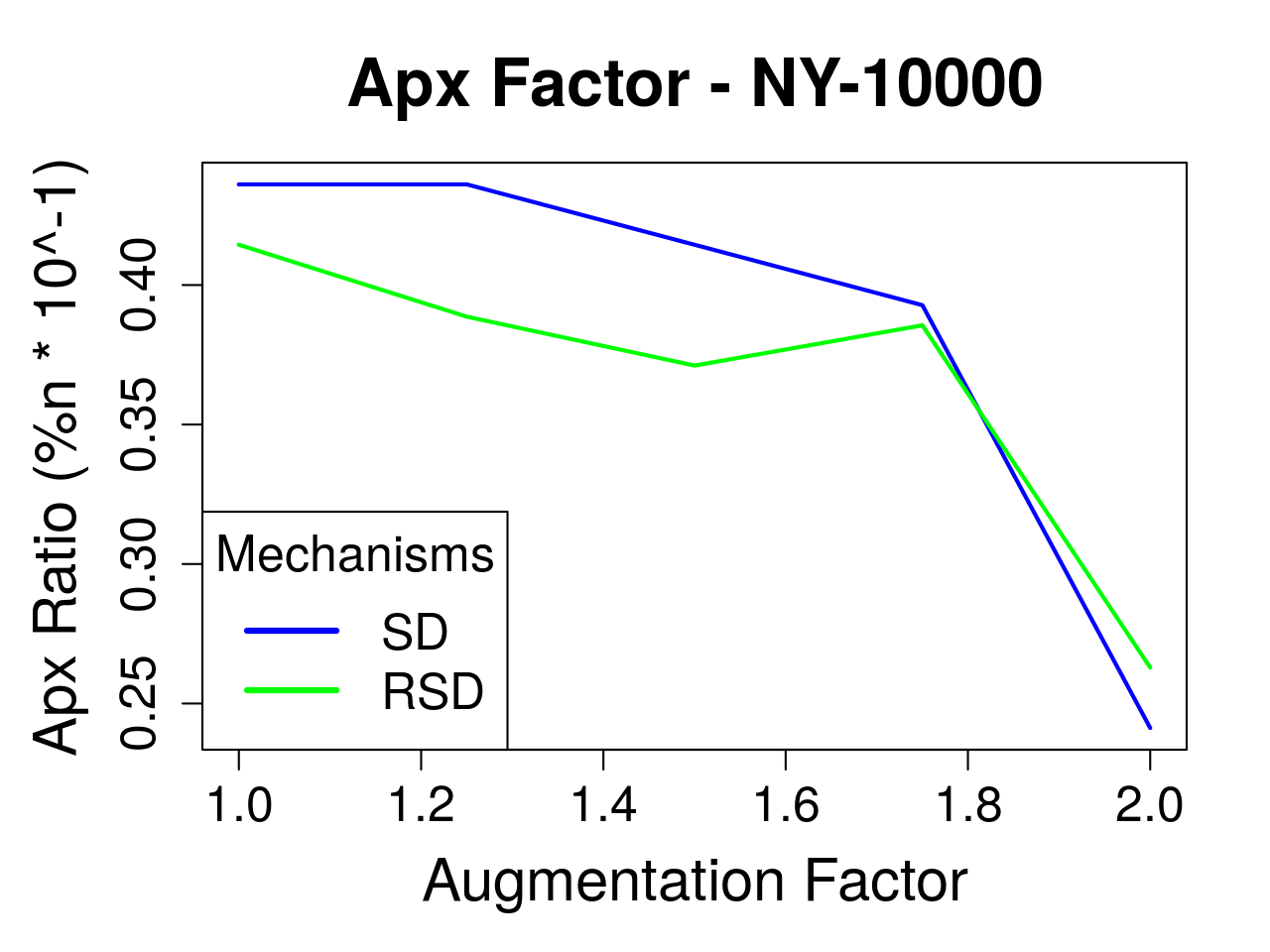}
\caption{Experimental results on NY-10000}\label{fig:NY-10000_results}
\end{minipage}
\vspace{-0.2cm}
\end{figure*}

\section{Experimental Results}\label{sec:experimental_results}
In this section we present the results of the experimental evaluation we conducted in order to assess %the empirical performances of the algorithms we studied in the previous sections.
%Our aim is to assess whether 
whether the theoretical inapproximability lower bounds impose a high approximation cost on real-life instances. In short, we will show that they do not.
In particular, we have measured the approximation ratio obtained by SD and RSD on three real-life graphs extracted from the DIMCAS 99 shortest path implementation challenge benchmark datasets \cite{DIMACS}.
In particular, Rome99 represents a large portion of the directed road network of the city of Rome, Italy, from 1999. The graph contains 3353 vertices and 8870 edges. Vertices correspond to intersections between roads and edges correspond to roads or road segments. 
NY-4000 and NY-10000 are two subgraphs extracted from NY-d, a larger distance graph (with 264,346 nodes and 733,846 edges) representing a large portion the road network infrastrucutre of New York City, USA.
The two graphs were obtained by taking a subset, respectively, of the first 4000 and 10000 nodes of the graph while ensuring that the connectivity was preserved by adding edges representing paths through nodes of the original graph not included in the subgraph.
In Table \ref{tab:graphs_stats} some statistics related to the structural characteristics of our test graphs are reported, where $\delta^+_{AVG}$ represents the average outdegree of a node (i.e. the average number of edges originating from a node) and $c_{AVG}$ is the average capacity of the outgoing edges of a node.
%\begin{table}
%\centering
%\caption{Structural characteristics of test graphs}\label{tab:graphs_stats}
%\label{my-label}
%\begin{tabular}{l|l|l|l|l|}
%\cline{2-5}
%                                 & Rome-99 & NY-4000 & \multicolumn{2}{l|}{NY-10000} \\ \hline
%\multicolumn{1}{|l|}{$|V|$}        & 3000 & 4000    & \multicolumn{2}{l|}{10000}    \\ \hline
%\multicolumn{1}{|l|}{$|E|$}        & 8859 & 10027   & \multicolumn{2}{l|}{312594}   \\ \hline
%\multicolumn{1}{|l|}{$\delta^+_{AVG}$} & 2.6  & 2.5     & \multicolumn{2}{l|}{31}       \\ \hline
%\multicolumn{1}{|l|}{$c_{AVG}$}   & 27.3 & 20.5    & \multicolumn{2}{l|}{30}       \\ \hline                                                    
%\end{tabular}
%\end{table}
%%
In our experimental assessment, we studied the variation of the approximation ratio of SD and RSD on the test graphs while varying the \emph{resource augmentation factor}.
The resource augmentation factor is the key parameter of the resource augmentation framework \cite{resource_augmentation}, a novel comparison framework where a truthful mechanism that allocates ``scarce resources'' is evaluated by its worst-case performance on an instance where such ``scarce resources'' are augmented, against the optimal mechanism on the same instance with the original amount of resources.
In \cite{resource_augmentation} it is argued  that this is a fairer comparison framework than the traditional approximation ratio, which compares the performance of a mechanism that is severely limited by the requirement of truthfulness to that of an omnipotent mechanism that operates under no restrictions and has access to the real inputs of the agents.
An equivalent resource augmentation framework is often also used in the analysis of online algorithms.
In the TAP scenario, the natural resource to be augmented is the capacity of the existing edges, modelled by  the augmentation factor $\gamma$, which in our framework is defined as the factor by which the average capacity of the edges departing from a node is multiplied, spreading the excess capacity evenly among the outgoing edges of the node.
More formally, if $c_{AVG}(v)$ is the average capacity of node $v$, then the augmented average capacity $c^{\gamma}_{AVG}(v) = \gamma \cdot c_{AVG}$, and the capacity of each outgoing edge is set as $\frac{c_{AVG}(v)}{\delta^+(v)}$, where $\delta^+(v)$ is the outdegree of $v$.
In our experiments we ranged the augmentation factor $\gamma$ in the interval $[1,2]$, which means increasing the initial capacity until it is doubled.
To run our experiments, we generated three separate populations of agent-origin-destination triplets, one population for each test graph, each comprising a number of triplets roughly equal to $1/3$ of the nodes of the graph. The size of the population of triplets was empirically tailored to let the competition for popular links arise without making the allocation problem unfeasible. 
For each agent-origin-destination triplet in the population, both the origin and the destination were independently drawn uniformly at random from the set of the nodes of the graph, with replacement (i.e. the same node can be the origin/destination of multiple triplets).

Figures \ref{fig:Rome_results}, \ref{fig:NY-4000_results} and \ref{fig:NY-10000_results} show the results of our experimental analysis, respectively on graph Rome99, NY-4000 and NY-10000.
In particular, the left hand side plot represents the absolute value of the social cost for the optimal mechanism,
%\footnote{Due to the intrinsic complexity of the optimal algorithm, in order to compute the optimal allocation we have relaxed the capacity constraints thus obtaining an infeasible allocation. However, this does not invalidate our analysis, as we obtain an upper bound on the approximation ratio of $SD$ and $RSD$ instead of the actual approximation ratio.}
expressed in kilometers, for SD and for RSD, whereas the right hand side plot represents the approximation ratio for SD and RSD.
From our experimental analysis we can see that the actual approximation ratio of both SD and RSD is much lower than the predicted theoretical worst-case approximation.
In particular, our experiments show that the approximation ratios of SD and RSD are quite similar and strongly $o(n)$ on the investigated road networks.
This is due to the fact that such theoretical approximation lower bounds rely on pathological instances that are quite unlikely to occur in real life graphs.
It is also worth noting the beneficial effect that augmenting the capacity of existing roads has on the approximation ratio: increasing the augmentation factor steadily decreases the approximation ratio on both Rome99 and NY-4000.
On the other hand a marked decrease is noticeable only if we increase the augmentation factor to 1.8 in the case of NY-10000.
This phenomenon is  due to the already reach topological structure of NY-10000, which necessitates less augmentation to yield good performances.

%\begin{figure*}[t!]
%\begin{minipage}{.32\columnwidth}
%\includegraphics[width=0.48\linewidth]{plots/rome_plot.png}
%\includegraphics[scale=.17]{plots/rome_plot_apx.png}
%\end{minipage}
%\begin{minipage}{.32\columnwidth}
%\includegraphics[width=0.48\linewidth]{plots/ny_4000_accuracy.png}
%\includegraphics[scale=.17]{plots/ny_4000_apx.png}
%\end{minipage}
%\begin{minipage}{.32\columnwidth}
%\includegraphics[width=0.48\linewidth]{plots/ny_10000_plot.png}
%\includegraphics[scale=.17]{plots/ny_10000_plot_apx.png}
%\end{minipage}
%\end{figure*}

% !TEX root = aaai19_traffic_assignment.tex
\section{Conclusions}
\label{Sec:conclusions}
In this paper we investigate the problem of strategyproof traffic assignment without monetary incentives.
We study two SP mechanism for our problem, namely Serial Dictatorship and its randomized counterpart Random Serial Dictatorships.
For deterministic mechanisms we prove that Serial Dictatorship is $2^n-1$ under some mild assumptions, and characterize Bipolar Serial Dictatorship as the only SP, Pareto optimal and non-bossy deterministic mechanism for our problem.
In the randomized case, we prove that Random Serial Dictatorship is $n$-approximate.
%, and a $\sqrt{n}$ approximation lower bound for mechanisms that are truthful in expectation.
Finally we assess the performance of Serial Dictatorship and Random Serial Dictatorship on real road network infrastructure, and show that they exhibit good approximation guarantees.
In particular, RSD is almost indistinguishable from SD, which means that the instances giving rise to the inapproximability results rarely occur in practice.

Note that our work is the first that addresses the problem of moneyless strategyproof traffic assignment. Although it ignores a number of properties that occur in real-world scenarios (e.g., dynamic network behavior, or asynchronous bid submissions), it still serves as a proof of concept for the existence of moneyless strategyproof assignment mechanisms.% Furthermore, it lays down the basis for future work. In particular, our next step is to investigate the dynamic version of the problem, in which:  (i) the capacity of the edges may vary over time, and thus, agents might revise their bids while they are travelling; and (ii) agents can join and leave the system dynamically, making the bidding process asynchronous. 
%Another potential future work is to extend the strategyproofness to the group level, where the goal is to prevent agents to strategically form coalitions and untruthfully submit their bids in order to maximise their coalition's benefit.

\bibliographystyle{splncs04}
\bibliography{biblio}

\begin{thebibliography}{10}
\providecommand{\url}[1]{\texttt{#1}}
\providecommand{\urlprefix}{URL }
\providecommand{\doi}[1]{https://doi.org/#1}

\bibitem{DIMACS}
{9th DIMACS Implementation Challenge - Shortest Paths}.
  \url{http://www.diag.uniroma1.it/challenge9/download.shtml} (1999)

\bibitem{appendix}
Anonymous: Appendix.
  \url{https://www.dropbox.com/s/8rid0ltqd4h6t2q/suppl_material_pricai_traffic_assignment.pdf?dl=0}
  (2018)

\bibitem{beckmann1956studies}
Beckmann, M., McGuire, C., Winsten, C.B.: Studies in the economics of
  transportation. Techn. report  (1956)

\bibitem{DBLP:journals/jet/BogomolnaiaDE05}
Bogomolnaia, A., Deb, R., Ehlers, L.: Strategy-proof assignment on the full
  preference domain. J. Economic Theory  \textbf{123}(2),  161--186 (2005).
  \doi{10.1016/j.jet.2004.05.004},
  \url{https://doi.org/10.1016/j.jet.2004.05.004}

\bibitem{Brenner2010}
Brenner, J., Sch{\"a}fer, G.: Online Cooperative Cost Sharing, pp. 252--263
  (2010)

\bibitem{resource_augmentation}
Caragiannis, I., Filos-Ratsikas, A., Frederiksen, S.K., Hansen, K.A., Tan, Z.:
  Truthful facility assignment with resource augmentation: An exact analysis of
  serial dictatorship. In: WINE'16. pp. 236--250 (2016)

\bibitem{conitzer2006failures}
Conitzer, V., Sandholm, T.: Failures of the {VCG} mechanism in combinatorial
  auctions and exchanges. In: AAMAS'06. pp. 521--528. ACM (2006)

\bibitem{coogan2015compartmental}
Coogan, S., Arcak, M.: A compartmental model for traffic networks and its
  dynamical behavior. IEEE Trans. on Automatic Control  \textbf{60}(10),
  2698--2703 (2015)

\bibitem{daganzo1994cell}
Daganzo, C.F.: The cell transmission model: A dynamic representation of highway
  traffic consistent with the hydrodynamic theory. Transportation Research Part
  B: Methodological  \textbf{28}(4),  269--287 (1994)

\bibitem{djahel2013adaptive}
Djahel, S., Salehie, M., Tal, I., Jamshidi, P.: Adaptive traffic management for
  secure and efficient emergency services in smart cities. In: PERCOM'13. pp.
  340--343. IEEE (2013)

\bibitem{Filos-RatsikasF014}
Filos{-}Ratsikas, A., Frederiksen, S.K.S., Zhang, J.: Social welfare in
  one-sided matchings: Random priority and beyond. In: SAGT'14. pp. 1--12
  (2014)

\bibitem{FordFulkerson1}
Ford, L.R., Fulkerson, D.R.: Constructing maximal dynamic flows from static
  flows. Oper. Res.  \textbf{6}(3),  419--433 (Jun 1958).
  \doi{10.1287/opre.6.3.419}, \url{http://dx.doi.org/10.1287/opre.6.3.419}

\bibitem{FordFulkersonBook}
Ford, L.R., Fulkerson, D.R.: Flows in Networks. Princeton University Press
  (1962), \url{http://www.jstor.org/stable/j.ctt183q0b4}

\bibitem{goodwin2004economic}
Goodwin, P.: The economic costs of road traffic congestion  (2004)

\bibitem{Kryzanowski2005}
Krzyzanowski, M., Kuna-Dibbert, B., Schneider, J.: Health effects of
  transport-related air pollution. WHO Regional Office Europe pp. 1--190 (2005)

\bibitem{leontiadis2011effectiveness}
Leontiadis, I., Marfia, G., Mack, D., Pau, G., Mascolo, C., Gerla, M.: On the
  effectiveness of an opportunistic traffic management system for vehicular
  networks. IEEE Trans. on Int. Transportation Sys.  \textbf{12}(4),
  1537--1548 (2011)

\bibitem{levin2017optimizing}
Levin, M.W., Fritz, H., Boyles, S.D.: On optimizing reservation-based
  intersection controls. IEEE Trans. on Int. Transportation Sys.
  \textbf{18}(3),  505--515 (2017)

\bibitem{Nesterov2003}
Nesterov, Y., de~Palma, A.: Stationary dynamic solutions in congested
  transportation networks: Summary and perspectives. Networks and Spatial
  Economics  \textbf{3}(3),  371--395 (Sep 2003).
  \doi{10.1023/A:1025350419398}, \url{https://doi.org/10.1023/A:1025350419398}

\bibitem{Nissim_2012}
Nissim, K., Smorodinsky, R., Tennenholtz, M.: Approximately optimal mechanism
  design via differential privacy. In: ITCS'12. pp. 203--213 (2012)

\bibitem{osorio2015urban}
Osorio, C., Nanduri, K.: Urban transportation emissions mitigation: Coupling
  high-resolution vehicular emissions and traffic models for traffic signal
  optimization. Transportation Research Part B: Methodological  \textbf{81},
  520--538 (2015)

\bibitem{ProcacciaT13}
Procaccia, A.D., Tennenholtz, M.: Approximate mechanism design without money.
  {ACM} TEAC.  \textbf{1}(4),  18:1--18:26 (2013).
  \doi{10.1145/2542174.2542175},
  \url{http://doi.acm.org/10.1145/2542174.2542175}

\bibitem{raphael2015goods}
Raphael, J., Maskell, S., Sklar, E.: From goods to traffic: first steps toward
  an auction-based traffic signal controller. In: PAAMS'15. pp. 187--198.
  Springer (2015)

\bibitem{skabardonis1997improved}
Skabardonis, A., Dowling, R.: Improved speed-flow relationships for planning
  applications. Transportation Research Record: Journal of the Transportation
  Research Board (1572),  18--23 (1997)

\bibitem{Stanfeld2003}
Stansfeld, S.A., Matheson, M.P.: Noise pollution: non-auditory effects on
  health. Oxford Journals, Medicine and Health, British Medical Bulletin
  \textbf{68(1)},  243--257 (2003)

\bibitem{Svensson1999}
Svensson, L.G.: Strategy-proof allocation of indivisible goods. Social Choice
  and Welfare  \textbf{16}(4),  557--567 (Sep 1999).
  \doi{10.1007/s003550050160}, \url{https://doi.org/10.1007/s003550050160}

\bibitem{vasirani2012market}
Vasirani, M., Ossowski, S.: A market-inspired approach for intersection
  management in urban road traffic networks. JAIR  \textbf{43},  621--659
  (2012)

\bibitem{Yao}
Yao, A.: Probabilistic computations: Toward a unified measure of complexity.
  In: FOCS. pp. 222--227 (1977)

\end{thebibliography}
\appendix 
\section{Appendix}
In the following we give the proofs of the theorems that were omitted in the main body of the paper due to space limitations.

\subsection{Time-Expanded Networks}
Let $G = (V, E)$ be a network with capacities $c$,
non-negative integral transit times $\tau$, and costs $w$ on the edges. 
For a given time horizon $T \in \mathbb{Z}>0$, the corresponding time-expanded network $G^T = (V^T, E^T )$ with capacities and costs on the edges
is defined as follows.
For each node $v \in V$ there are $T$ copies $v_0 , v_1 ,\ldots, v_{T -1}$, that is,
$V^{T} = \{v_t | v \in V, t \in \{0, 1, \ldots ,T - 1\}\}$ .
For each edge $e = (v, w) \in E$, there are $T - \tau(e)$ copies $e_0, e_1,\ldots , e_{T -1-\tau(e)}$ where edge $e_t$ connects
node $v_t$ to node $w_{t + \tau(e)}$ . Edge $e_t$ has capacity $c(e_t) = c(e)$  and cost $w(e_t) = w(e)$. 
Moreover, $E_T$ contains holdover edges $(v_t , v_{t+1} )$ for $v \in V$ and $t = 0,\ldots, T - 2$.
The capacity of holdover edges is infinite and they have zero cost.

\subsection{Omitted Theorems}
\begin{theorem}
Let $S=M(\D)$ be a traffic assignment and let $R_i(S)$ be the set of reaction available to $a_i$ at $S$.
If $S$ is Pareto optimal, then $M_i(\D)\in \arg\min_{P\in R_i(S)}w(P)$.
\end{theorem}
\begin{proof}
If $M_i(\D)\notin \arg\min_{P\in R_i(S)}w(P)$, then there must exist a reaction $r'_i\in R_i$ that is strictly better than $M_i(\D)$ for agent $a_i$, i.e. $c(r'_i)<c(M_i(\D))$.
Then, a route assignment $M'$ such that $M'_j = M_j$ for all $j\neq i$ and $M'_i = r_i$ is still feasible.  Since $cost_j(M')=cost_j(M')$ for all $j\neq i$, and $cost_i(M')<cost_i(M)$, $M$ is not Pareto optimal. 
\end{proof}
\begin{theorem}\label{thm:deterministic_LB}
There is no $\alpha$-approximate deterministic SP mechanism for the traffic assignment problem with $\alpha<3 - \varepsilon$, for any $\varepsilon>0$.
\end{theorem}
\begin{proof}
Given $\varepsilon>0$, consider the graph depicted in Figure \ref{fig:LB_instance_2}, where the labels on the edges represent their capacity (red) and length (black). The instance we consider has two agents $A= \{a_1,a_2\}$, both initially located at node $A$, whose intended destination is $D$ and $G$, respectively (namely, $\D=(D,G)$).
The length of the path $(B,C)$ is $K = \min\{2, \frac{10-4\varepsilon}{\varepsilon}\}$ and the length of path $(F,E)$ is $K-1$. 
\begin{figure}[h]
\centering
\includegraphics[width=.55\linewidth]{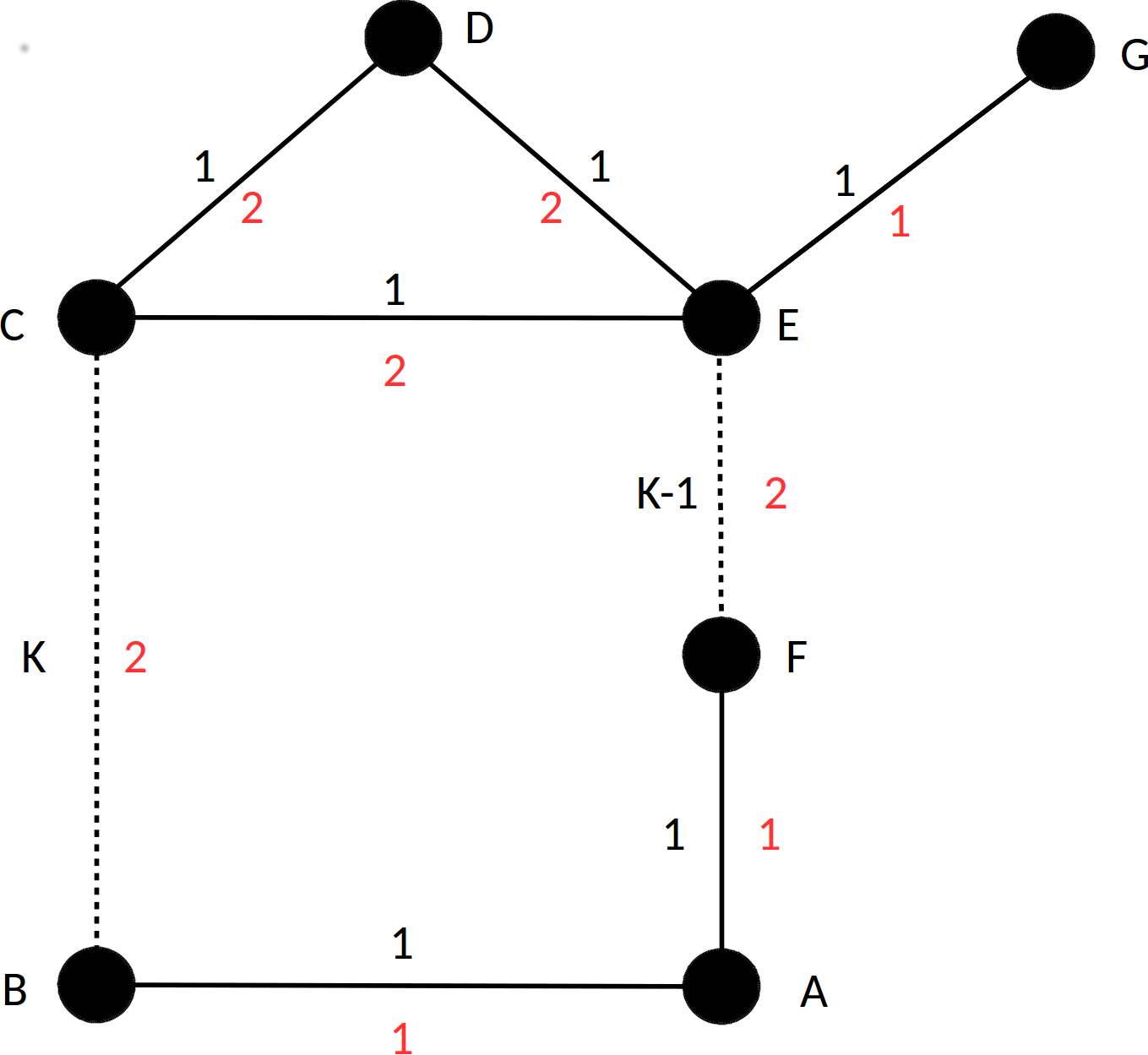}
\caption{Lower bound instance}\label{fig:LB_instance_2}
\end{figure}
Let us consider a generic $\alpha$-approximate mechanism $M$.
Assume by contradiction that $M$ is strategyproof and $\alpha$-approximate with $\alpha<3-\varepsilon$. %, where $\varepsilon = \frac{10}{k+4}$. 
The instance has two Pareto optimal solutions, depending on which player is allocated the edge $(A,F)$ (note that only one agent at a time can use edge $(A,F)$ as its capacity is 1). The optimal allocation $\mathbf{P}^*=OPT(\D)$ is $P^*_1 = (A,B,C,D)$ and $P^*_2 = (A,F,E,G)$, $cost_1(\mathbf{P}^*,D) = K+2$ and $cost_2(\mathbf{P}^*,G) = K+1$, and $SC(\mathbf{P}^*,\D)=2K+3$. The second best solution is $P_1 = (A,F,E,D)$ and $P_2 = (A,B,C,E,G)$. We note that $cost_1(P_1,D_1)=K+1$ and $cost_2(P_2,D_2)=K+3$, for a social cost of $SC(M(\D),\D) = 2K+4$. %; observe that this solution has an approximation ratio lower than $3-\varepsilon}$ when $K\geq \frac{31+\sqrt{1761}}{8}$.
We are going to prove that, regardless of the solution $M$ returns on this instance, there is another instance where to maintain SP, $M$ achieves an approximation not better than $3-\varepsilon$, a contradiction.
Let us assume first that $M$ returns the optimal allocation.
If agent $a_1$ declares $D'_1=F$ instead of her true type, by SP, $M$ cannot allocate the edge $(A, F)$ to $a_1$. In fact, assume for the sake of contradiction that $M(D_1', D_2)$ allocates $(A,F)$ to $a_1$.
Then $a_2$ is allocated path $(A,B,C,E,G)$ and $a_1$ can use path $(A,F,E,D)$ and reach her true destination, thus having: 
 $$cost_1(M(D'_1,D_2),D_1)=K+1<cost_1(M(\D),D_1)=K+2.$$ 
Therefore, $M(D_1', D_2)$ must return $P'_1 = (A,B,C,E,F)$ and $P'_2 = (A,F,E,G)$, with $SC((P'_1,P'_2), (D_1', D_2)) = 3K+2$, whilst the optimum for this instance is $OPT_1(D'_1,D_2) = (A,F)$ and $OPT_2(D'_1,D_2) =(A,B,C,E,G)$. The social cost of the optimum is then $SC(OPT(D'_1,D_2),(D'_1,D_2)) = K+4$.
Therefore $M$ has an approximation ratio higher than $3-\varepsilon$.
Let us now suppose that $M(\D)$ returns $M_1(\D) = P_1$ and $M_2(\D) = P_2$.
In this case, consider the case that agent $a_2$ reports $D'_2=F$ instead of her true type.
The optimal allocation is $OPT_1(D_1,D_2') = (A,B,C,D)$ and $OPT_2(D_1,D_2') = (A,F)$, and $SC(OPT(D'_2,D_1),(D'_2,D_1)) = K+3$.
As before, this allocation is not strategyproof as:
\begin{eqnarray*}
cost_2(OPT_2(D_1, D'_2),D_2)= K+1 \\ 
< cost_2(M_2(D_1,D'_2),D_2)=K+3 
\end{eqnarray*}
(i.e., agent $a_2$ can use the route $(A,F,E,G)$ to reach her true destination).
As above, one can easily check that in this case the best (in terms of approximation ratio) strategyproof allocation is $P'_1=(A,F,E,D)$ and $P'_2=(A,B,C,E,F)$, with a cost of $SC((P'_1,P'_2),(D_1,D_2')) = 3K+2$.
This solution has an approximation ratio higher than $3-\varepsilon$.
\end{proof}
\begin{theorem}
If $G$ is $K$-edge-connected\footnote{A graph is $K$-edge-connected if removing at most $K$ edges from it does not disconnect the graph.}, mechanism Serial Dictator is feasible for $K$ agents.
\end{theorem}
\begin{proof}
%Given a fixed ordering $a_1\prec \ldots \prec a_k$ over the agents, Serial Dictator works by assigning to each agent $a_1$ is most preferred path, and, for all $j>1$ assigning agent $a_j$ his most preferred path among the ones remaining after agents $\{a_1,\ldots,a_{j-1}\}$.
If the graph is $K$-edge connected, the allocation returned by the Serial Dictator will always be feasible, (i.e. paths assigned to different agents will not overlap and there is always an assignable path for each agents).
This follows from the fact that in a $K$-edge-connected graph there are at least $K$ edge disjoint paths between any pair of nodes.
\end{proof}

\setcounter{theorem}{4}

\begin{theorem}\label{thm:tightness_SD}
Under the DoCP assumption, the bound of Theorem 4 is tight.
\end{theorem}
\begin{proof}
Let us consider the instance in Figure \ref{fig:sd_tight_instance}, where there are $n$ nodes $v_1,\ldots, v_n$ and $n$ agents $A=\{a_1,\ldots, a_n\}$ such that agent $a_i$ is initially located at node $v_i$. All agents want to reach the same destination $D$.
Each link has capacity $1$.
Agent $a_1$ has two paths to her destination $D$: one direct path that costs $1+\varepsilon$ (where $\varepsilon\ll 1$ is a small constant) and a path costing $1$ that goes through the node agent $a_2$ is initially located on.
Each agent $a_i$, for $i=2,\dots,n-1$ has two paths: one direct path that costs $\varepsilon$ and a path costing $2^{i-1}$ that goes through the node  agent $a_{i+1}$ is initially located on.
Agent $a_n$ has two direct paths, costing $\varepsilon$ and $2^{n-1}$ respectively.
The optimal traffic assignment assigns agent $a_1$ to the path that costs $1+\varepsilon$ and the other agents to the path costing $\varepsilon$, and has a cost of $1+\varepsilon n$.
Let us consider ordering $a_1\prec a_2 \prec \ldots \prec a_n$.
On this ordering, mechanism SD assigns agent $a_1$ the path costing $1$, and to each agent $a_i$, for $i=2,\ldots, a_n$ the path costing $2^{i-1}$ for a total cost of $\sum_{i=0}^{n-1} 2^i = 2^n-1$.
For $\varepsilon$ close to $0$, the approximation ratio of SD on the instance depicted in Figure \ref{fig:sd_tight_instance} is hence close to $2^n-1$.
\begin{figure}[h]
\centering
\includegraphics[width=.6\linewidth]{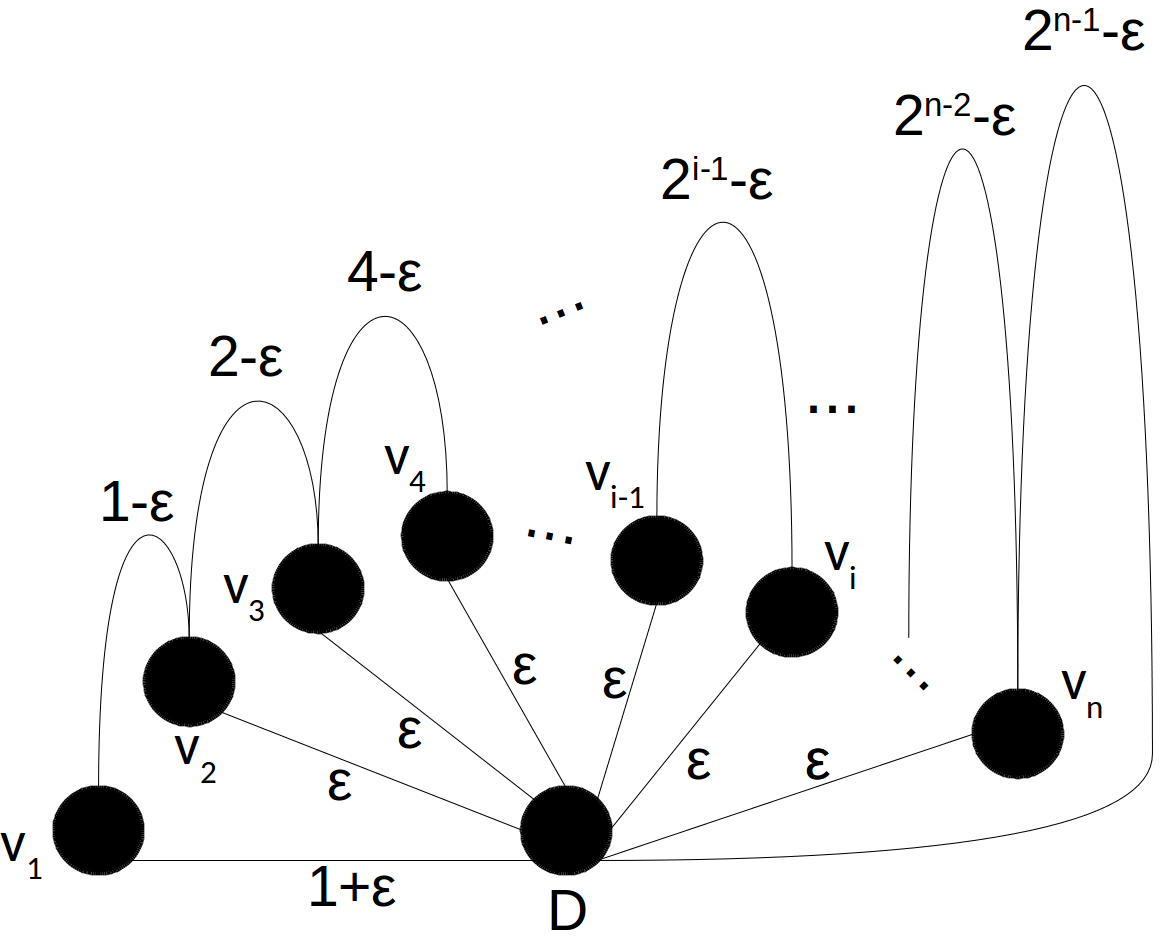}
\caption{Tight instance for SD}\label{fig:sd_tight_instance}
\end{figure}
\end{proof}
\setcounter{theorem}{5}

\begin{theorem}\label{thm:reduction}
A traffic allocation mechanism for \TAPStar is Pareto-optimal, SP and non-bossy if and only if it is a Bi-polar Serially Dictatorial Rule.
\end{theorem}
\begin{proof}
We will reduce an instance of the problem of \emph{assigning indivisible objects} with general ordinal preferences 
\cite{DBLP:journals/jet/BogomolnaiaDE05} (AIO for short) to \TAPStar{}.
%In the ordinal version of TAP, agents submit a preference relation over the set of paths, instead of a single node representing their destination.
%We will show that this is without loss of generality from the point of view of characterizing cardinal mechanisms.
An instance of AIO is composed of a set of objects $X = \{x_1,\ldots ,x_m\}$ that have to be assigned to a set of agents $A = \{a_1,\ldots ,a_n\}$, such that every agent receives at most one object and no agent is left without an object if there are objects still available.
Agents have ordinal general preferences $\succeq _i$, where $x \succeq_i y$ for $x,y\in X$ means that agent $i$ (weakly) prefers object $x$ to object $y$.
From an instance of AIO, we can build an instance of \TAPStar{} as follows. 
\TAPStar{} has the same set of agents $A$ as AIO.
Graph $G$ of \TAPStar{} has a node $O$ such that $O_i=O$ for all $a_i \in A$.
For every object $x_j\in X$ we construct in $G$ a node $v_j$ and an edge $(O,v_j)$ such that $c(O,v_j)=1$ and $w(O,v_j) = \varepsilon$ for $0<\varepsilon\ll 1$.
Let $\Psi$ be the set of all possible preference relations over $X$. 
We construct\footnote{We note that, although $|\Psi|$ can be exponential in $m$, it is always finite. We remark that graphs of exponential size are not an issue here since the characterization we are proving in this theorem does not rely on computational efficiency.} $|\Psi|$ destination nodes $D_k$, one for each preference relation $\succeq \in \Psi$ and for each $k\in{1,\ldots, |\Psi|}$.
For each $j \in \{1,\ldots, m \}$ we add an edge $(v_j,D_k)$ having capacity 1 and weight $w(v_j,D_k)$ equal to the \emph{ranking}\footnote{
The alternatives in $X$ can be partitioned in subsets ${X_1,\ldots, X_\ell,\ldots}$ such that any two elements $x_1,x_2\in X_\ell$ are indifferent according to $\succeq$ and, for any $x_1\in X_\ell$ and $x_2 \in X_{\ell+1}$, $x_1$ is strictly preferred to $x_2$ according to $\succeq$. Then $\ell$ is the ranking of $x\in X_\ell$.} of $x_j$ according to $\succeq$.
Figure \ref{fig:reduction_example} gives an example of the reduction applied to an AIO game with $A=\{a_1,a_2,a_3\}$, $X=\{x_1,x_2,x_3\}$ and $\Psi$ being the set of all possible \emph{linear orderings} over $X$. The labels on the edges of the graph of Figure \ref{fig:reduction_example} represent the costs of the edges, whereas all capacities are set to 1.
%\begin{table}[]
%\centering
%\caption{Mapping between elements of $\Psi$ and the destinations nodes of $G$}
%\label{my-label}
%\begin{tabular}{|l|l|l|l|}
%\hline
%$D_1$ & $x_1$ & $x_2$ & $x_3$ \\ \hline
%$D_2$ & $x_1$ & $x_3$ & $x_2$ \\ \hline
%$D_3$ & $x_2$ & $x_1$ & $x_3$ \\ \hline
%$D_4$ & $x_2$ & $x_3$ & $x_1$ \\ \hline
%$D_5$ & $x_3$ & $x_1$ & $x_2$ \\ \hline
%$D_6$ & $x_3$ & $x_2$ & $x_1$ \\ \hline
%\end{tabular}
%\end{table} 
\begin{table}[]
\centering
\caption{Mapping elements of $\Psi$ to destination nodes of $G$}
\label{my-label}
\begin{tabular}{|l|l|}
\hline
$D_1$ & $x_1\prec x_2\prec x_3$ \\ \hline
$D_2$ & $x_1\prec x_3\prec x_2$ \\ \hline
$D_3$ & $x_2\prec x_1\prec x_3$ \\ \hline
$D_4$ & $x_2\prec x_3\prec x_1$ \\ \hline
$D_5$ & $x_3\prec x_1\prec x_2$ \\ \hline
$D_6$ & $x_3\prec x_2\prec x_1$ \\ \hline
\end{tabular}
\end{table} 

By construction, the following hold: (\emph{i}) any path allocation on $G$ must include all the edges $(O,v_j)$; (\emph{ii}) any edge $(O,v_j)$ is used by at most one path; and (\emph{iii}) only one agent can be assigned any given edge $(O,v_j)$ due to the capacity constraint.
We can now easily transform a path allocation for the so-constructed \TAPStar{} problem to an allocation of objects to agents in the AIO problem, and vice versa.
Indeed, let $P_i$ be the path assigned to agent $a_i$ in the \TAPStar{} problem. If $P_i$ contains edge $(O,v_j)$ we allocate object $x_j$ to agent $a_i$ in the AIO instance, and vice versa from an allocation for the AOI problem to an allocation for the \TAPStar{} problem.
%by simply constructing a node $v_j$ and an edge $(S,v_j)\in E_S$ in $G$ for each object $x_j\in X$ for each edge.
In \cite{DBLP:journals/jet/BogomolnaiaDE05} it is proved that BSD is the only Pareto optimal, SP and non-bossy mechanism for AIO.
This characterization trivially transfers to \TAPStar{} due to the reduction sketched above.
Indeed, let us suppose that there exists an SP, Pareto optimal and non-bossy algorithm for TAP.
Such algorithm would be SP, Pareto optimal and non-bossy for the AIO instance as well.
%It is easy to see that the characterization obtained for the ordinal version of cardinal \TAPStar{} too.
%Indeed, let us assume for the sake of contradiction that there exists a cardinal mechanism $M$ for the TAP problem that is not BSD but it is SP, Pareto optimal and non-bossy.
%We can easily build an ordinal mechanism $M'$ from $M$ as follows.
%Mechanism $M'$  as follows. 
%For every agent $i$, $M'$ computes makes up a cost function for $i$ that is consistent with $\succeq_i$.
%Path $cost_i(P) = 1$ for all $P \in \mathcal{P}^i_1 = max_{\succeq_i} \{{\mathcal{P}}\}$ and $cost_i(P) = 2$ for all $P\in \mathcal{P}^i_2=max_{\succeq_i} \{\mathcal{P}\setminus \mathcal{P}^i_1\}$ and $cost_i(P) = k$ if $P\in \mathcal{P}^i_2=max_{\succeq_i} \{\mathcal{P}\setminus \mathcal{P}^i_{k-1}\}$ .
%Subsequently, $M'$ feeds the cardinal preferences to mechanism $M$ and returns $M$'s output on such input.
\begin{figure}[h]
\centering
\includegraphics[width=.6\linewidth]{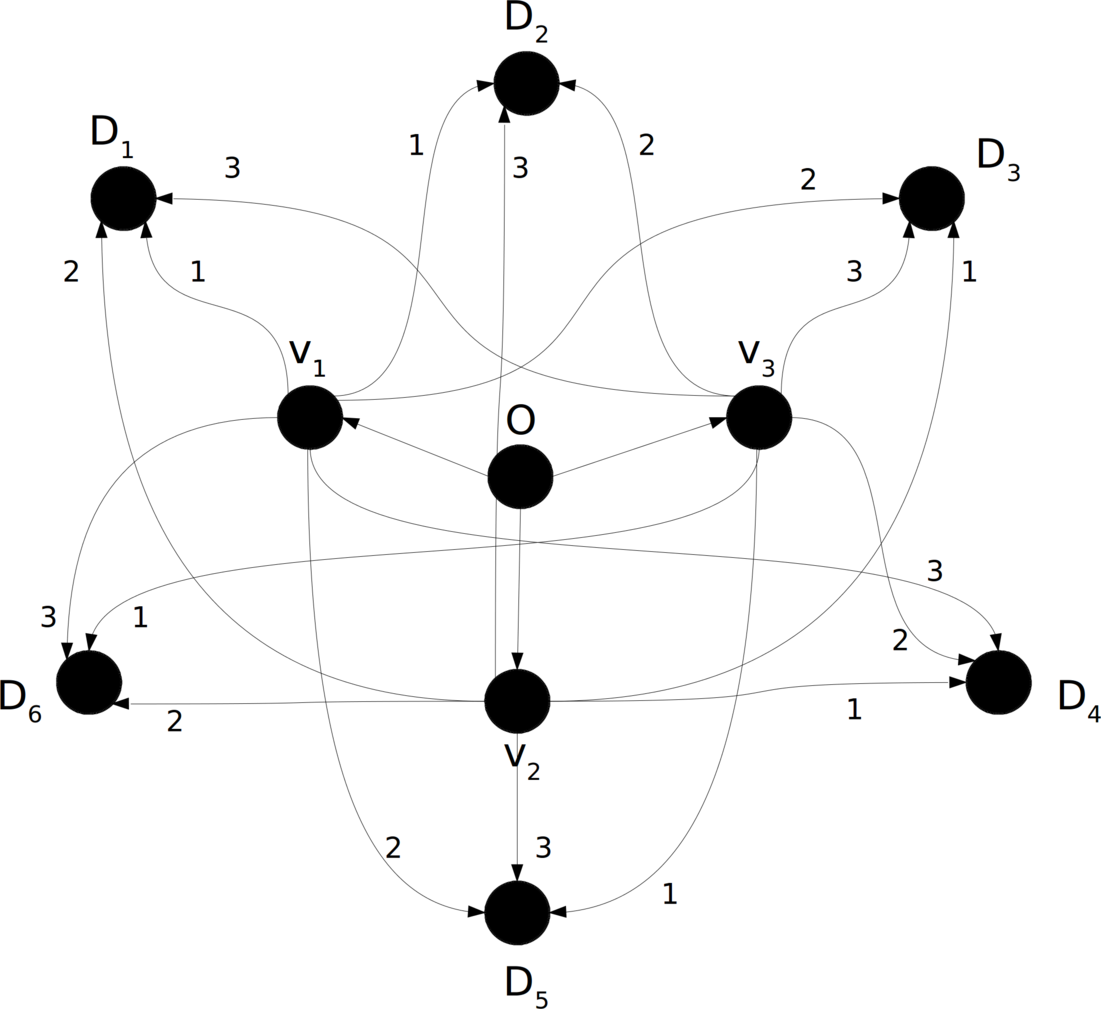}
\caption{Reduction example}\label{fig:reduction_example}
\end{figure}
\end{proof}

\begin{lemma}
BSD cannot achieve an approximation ratio lower than $\Omega(2^n)$ for TAP.
\end{lemma}
\begin{proof}
We are going to show an instance of TAP where BSD has an approximation ratio of $\Omega(2^n)$.
Let us take the instance of Figure \ref{fig:sd_tight_instance} and let us consider the ordering $\{a_1,a_2\}\prec a_3\prec \ldots \prec a_n$.
Let us consider $X_1 = \{v_2,D\}$ and $X_2 = E\setminus X_1$.
The so-defined BSD mechanism, on input the instance of figure \ref{fig:sd_tight_instance} would always execute SD with ordering $a_1 \prec a_2 \prec \ldots, a_n$. We know from Theorem \ref{thm:tightness_SD} that under this ordering the approximation ratio of SD is $\Omega(2^n)$.
\end{proof}

\begin{theorem}\label{thm:randomizedLB}
There is no $\alpha$-approximate randomized universally truthful mechanism for the traffic assignment problem with $\alpha<11/10$.
\end{theorem}
\begin{proof}
Our approach is based on Yao's minimax principle \cite{Yao}. In our context, this principle states that the approximation ratio of the best universally truthful randomized mechanism is equal to the approximation ratio of the best deterministic truthful mechanism under a worst-case input distribution. Accordingly, we exhibit a probability distribution over input instances for which any deterministic truthful mechanism cannot attain an approximation guarantee better than $11/10$.
The two instances are taken from the proof of Theorem \ref{thm:deterministic_LB}, where we set $K=2$. Specifically, we consider the instance in Figure \ref{fig:LB_instance_2}, that we name $I$, and the very same instance where agent $a_1$ reports $F$; we call this instance $I'$. We consider a probability distribution over $I$ and $I'$ that returns $I$ with probability $\lambda$ and $I'$ with the remaining probability $1-\lambda$, where $\lambda=2/3$.  The expected value of the optimum will then be $(\lambda+1)K+4-\lambda=20/3$.
Let $M$ be a SP deterministic mechanism. From the arguments in the proof of the theorem above, we know that $M$ must assign the edge $(A,F)$ to the same agent in both instances $I$ and $I'$. If $M$ allocates $(A,F)$ to agent $a_1$ in both the instances then its expected social cost will be $(\lambda+1)K+4=22/3$ for an approximation ratio of $11/10$. If instead $M$ allocates $(A,F)$ to agent $a_2$ in both the instances then the expected social cost of the mechanism will be $(3-\lambda)K+\lambda+2=22/3$; the approximation ratio of $M$ would then be $11/10$.
\end{proof}

\setcounter{theorem}{8}
\begin{theorem}
The approximation ratio of RSD is $\Omega(n)$.
\end{theorem}
\begin{proof}
The proof uses the same construction as the instance of Figure \ref{fig:sd_tight_instance}, with $k<n$ nodes.
One agent is initially located at node $v_1$, whereas $1+2\cdot3^{i-1}$ agents are initially located at node $v_i$, for $i=2,\ldots, k-1$ .
With a little abuse of notation, let $|v_i|$ denote the number of agents initially located at node $v_i$, and let $n_i = \sum_{\ell=0}^i |v_\ell|$.
Edges $(v_1,D)$ and $(v_1,v_2)$ have capacity $1$, whereas edges $(v_i,D)$ and $(v_i,v_{i+1})$ have capacity $1+2\cdot3^{i-1}$ for $i>1$.
Let $a_1 \prec\ldots \prec a_n$ be an ordering over the agents.
We will be interested in orderings that possess the \emph{chain of levels} property, namely for all $i=1,2,\ldots, k-1$ at least one agent located at node $i$ appears after all agents of levels $0,1,\ldots i-1$.
The property of a chain of levels ordering with respect to the instance of Figure \ref{fig:sd_tight_instance} is that it forces at least one agent located at node $v_i$, for all $i = 1,\ldots, k-1$ to use the path $P = (v_i,v_{i+1},D)$, at a cost of $2^{i-1}$ for the agents, and an overall social cost of $\sum_{i=1}^{k-1}2^{i-1}=2^k-1 > 2^{k-1}$.

We argue that the probability that a chain of levels ordering is chosen by RSD is $\Pi_{i=1}^{k-1} \left(1-\frac{n_{i-1}}{n_i}\right)$.
Indeed, we can look at the process of randomly generating an ordering as follows. First an ordering for the agents located at each node is uniformly generated at random.
Then orderings of agents of consecutive nodes are merged together in lexicographic order.
In particular, we start merging the orderings of nodes $v_1$ and $v_2$.
There are $\binom{|v_2|+|v_1|}{v_1} = \binom{n_2}{n_1}$ such orderings, whereas there are $\binom{n_2-1}{n_1}$ orderings where one agent located at node $v_2$ follows all the agents located at node $v_1$. The partial ordering obtained so far is randomly merged with the ordering of agents at node $v_3$ and the procedure continues until the partial ordering is complete. When merging agents at node $v_i$ with the current partial ordering, we note that there are $\binom{n_i}{n_{i-1}}$ possible orderings, and $\binom{n_i-1}{n_{i-1}}$ where for all $\ell = 1,\ldots, i$, one agent located at node $v_\ell$ follows all the agents located at node $v_{\ell-1}$ in the ordering (i.e., fix one agent from node $v_\ell$ in the last position and compute all possible orderings of the other agents).
Hence, the probability of one agent at node $v_\ell$ appearing after all agents at node $v_{\ell-1}$ is $\binom{n_i-1}{n_{i-1}}/\binom{n_i}{n_{i-1}} = \left(1-\frac{n_{i-1}}{n_i}\right)$.
Since the random orderings generated at each stage are independent, the probability that for all $i = 1, 2, \ldots, k-1$ at least one agent at node $v_i$ appears after all agents at node $v_{i-1}$ in a random ordering is $\Pi_{i=1}^{k-1}\left(1-\frac{n_{i-1}}{n_i}\right)$.
Hence, the probability that a chain of levels ordering is chosen by RSD for the instance of Figure \ref{fig:sd_tight_instance} is $(2/3)^{k-1}$.

Finally, the expected cost of RSD is at least $(4/3)^{k-1} = n^{\log_3(4/3)}\approx n^{0.262}$. Since the optimal allocation costs $1+\epsilon \cdot n$, the approximation ratio is $\Omega(n)$ for $\epsilon$ close to $0$.
\end{proof}
\end{document}